\documentclass[sigconf, 9pt]{acmart}

\renewcommand\footnotetextcopyrightpermission[1]{}

    \usepackage{indentfirst}
	\usepackage{graphicx} 
	\usepackage{balance}
	\usepackage{amsmath}
	\usepackage{enumitem}  
	\usepackage{multicol}
        \usepackage{colortbl}
        \usepackage[title]{appendix}
	\usepackage{multirow, rotating, wasysym, url}
	\usepackage[tight]{subfigure}
	\usepackage[ruled,linesnumbered,vlined]{algorithm2e}
	\usepackage{hyperref}
	\usepackage{color}
	\usepackage{soul}
	\usepackage{bbm}
	\usepackage[margin=0.05in]{caption}
	\usepackage{enumerate}
	\usepackage{framed}
	\usepackage{lipsum}
 \usepackage{diagbox}
	\usepackage[normalem]{ulem}
        \usepackage{booktabs}
        \usepackage{wrapfig}

	\newcommand{\presec}{\vspace{-0.0cm}}
	\newcommand{\postsec}{\vspace{-0.0cm}}
	
	\newcommand{\presub}{\vspace{-0.0cm}}
	\newcommand{\postsub}{\vspace{-0.0cm}}

	\newcommand{\m}{\mathcal}
    \newcommand{\bb}{\textbf}

	\newtheorem*{Rmk}{Remark}
	\newtheorem{Cor}{Corollary}[section]
	\mathchardef\Gamma="0100 \mathchardef\Delta="0101
\mathchardef\Theta="0102 \mathchardef\Lambda="0103
\mathchardef\Xi="0104 \mathchardef\Pi="0105
\mathchardef\Sigma="0106 \mathchardef\Upsilon="0107
\mathchardef\Phi="0108 \mathchardef\Psi="0109
\mathchardef\Omega="010A

\newcommand{\outline}[1]{}

\usepackage{xspace}
\usepackage{url}
\usepackage{graphicx}
\usepackage{latexsym}
\usepackage{amsfonts}
\usepackage{psfrag}
\usepackage{wrapfig}
\usepackage{comment}
\usepackage{alltt}
\usepackage{color}

\newtheorem*{defn*}{Definition}

\newcommand{\Comment}[1]{}

	\definecolor{zyd}{RGB}{50,50,200}

	\definecolor{yt}{RGB}{0,166,0}
	
	\definecolor{zz}{RGB}{0,200,200}
	
	\definecolor{hwc}{RGB}{200,0,200}
	
	\definecolor{zyk}{RGB}{200,50,50}

	\definecolor{remove}{RGB}{200,200,200}
	\newcommand{\remove}[1]{}
	
	\newcommand{\algoname}{ChainedFilter}
        
\begin{document}
\settopmatter{printfolios=true}
\title{\algoname{}: Combining Membership Filters by Chain Rule}
\author{Haoyu Li}
\affiliation{\institution{UT Austin\\ Peking University}\country{}}
\settopmatter{authorsperrow=4}
\author{Liuhui Wang}
\affiliation{\institution{University of Pennsylvania}\country{}}

\author{Qizhi Chen}
\affiliation{\institution{Peking University}\country{}}

\author{Jianan Ji}
\affiliation{\institution{Peking University}\country{}}

\author{Yuhan Wu}
\affiliation{\institution{Peking University}\country{}}

\author{Yikai Zhao}
\affiliation{\institution{Peking University}\country{}}

\author{Tong Yang}
\authornote{Co-corresponding authors.}
\affiliation{\institution{Peking University}\country{}}

\author{Aditya Akella}
\authornotemark[1]
\affiliation{\institution{UT Austin}\country{}}
        \presec
\begin{abstract}
\postsec
Membership (membership query/membership testing) is a fundamental problem across databases, networks and security. However, previous research has primarily focused on either approximate solutions, such as Bloom Filters, or exact methods, like perfect hashing and dictionaries, without attempting to develop a an integral theory. In this paper, we propose a unified and complete theory, namely chain rule, for general membership problems, which encompasses both approximate and exact membership as extreme cases. Building upon the chain rule, we introduce a straightforward yet versatile algorithm framework, namely \algoname{}, to combine different elementary filters without losing information. Our evaluation results demonstrate that \algoname{} performs well in many applications: (1) it requires only 26\% additional space over the theoretical lower bound for implicit static dictionary, (2) it requires only 0.22 additional bit per item over the theoretical lower bound for lossless data compression, (3) it reduces up to 31\% external memory access than raw Cuckoo Hashing, (4) it reduces up to 36\% P99 tail point query latency than Bloom Filter under the same space cost in RocksDB database, and (5) it reduces up to 99.1\% filter space than original Learned Bloom Filter.
\end{abstract}
        
         \maketitle

	\presec
\section{Introduction}
\label{sec:intro}


Membership has been a fundamental problem for over fifty years, playing a significant role in databases \cite{tarkoma2011theory}, networks \cite{broder2004network} and security \cite{geravand2013bloom}. 
For instance, LSM-Tree based storage engines employ Bloom Filter \cite{chang2008bigtable,dayan2018optimal,matsunobu2020myrocks} to accelerate K-V stores; Routers and switches leverage Bloom Filter to classify, forward, and drop network packets \cite{dharmapurikar2006fast,li2011scalable,reviriego2020cuckoo}; Bitcoin miners use Invertable Bloom Lookup Tables (IBLT) to reduce the amount of information for block propagation and reconciliation \cite{goodrich2011invertible,ozisik2019graphene,imtiaz2019churn}.

Given a universe $\m{U}$ and a subset $\m{S}$ of $n$ items, \textbf{membership} aims to determine whether an item $x \in \m{U}$ is in $\m{S}$. Specially, a membership algorithm must say ``yes'' if $x$ is in $\mathcal{S}$ and may produce a small \bb{false positive rate} $\epsilon$ if $x$ is not in $\mathcal{S}$. In this paper, {we further define negative-positive ratio $\lambda:=|\m{U}\backslash \m{S}|/|\m{S}|$\footnote{The number of items not in $\m{S}$ divide the number of items in $\m{S}$.} and divide membership problems into three categories (\bb{Figure \ref{pic::taxonomy}})}: \bb{approximate membership}, where $\epsilon\neq0$ and $\lambda\to +\infty$; \bb{exact membership}, where $\epsilon=0$ and $\lambda<+\infty$; and \bb{general membership}, where $\epsilon\neq 0$ and $\lambda<+\infty$. According to our taxonomy, both approximate and the exact memberships are the extreme cases of general membership. But in history, the membership problems were not classified in this manner. In 1978, \cite{carter1978exact} proposed separate space lower bounds for approximate and exact memberships. Since the two expressions have significantly different forms, in the following decades, people regarded the approximate and the exact memberships as two separate research directions, but never tried to unify them to develop an integral theory for general membership problems, until this work.

\begin{figure}[h!tbp]
  \centering
   \includegraphics[width=0.36\textwidth]{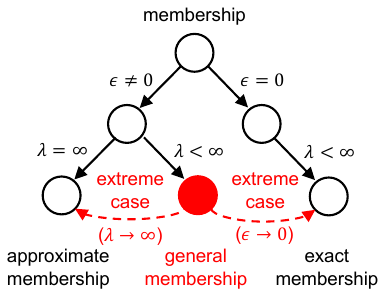}
\caption{Our taxonomy.}
\label{pic::taxonomy}
\end{figure}

Build on \cite{carter1978exact}, \textbf{we develop a unified space lower bound $nf(\epsilon,\lambda)$ for general membership problems, which encompasses the prior theoretical results as extreme cases}. Given a false positive rate $\epsilon$ and a negative-positive ratio $\lambda$, we have the following expression (where $H(\cdot)$ denotes Shannon's entropy, ignore $o(1)$ terms):
\begin{equation*}
    \begin{cases}
    f\left(\epsilon, \lambda\right) = \underbrace{f\left(\epsilon',\lambda\right)}_{\textbf{First stage}}+\underbrace{f\left(\epsilon/{\epsilon'}, \epsilon'\lambda\right)}_{\textbf{Second stage}},\forall \epsilon'\in[\epsilon,1] \text{ (Chain rule)};\\
        f\left(0, \lambda\right)=\left(\lambda+1\right)H\left(\frac{1}{\lambda+1}\right)\text{ (Exact membership bound \text{\cite{carter1978exact}})}.
    \end{cases}
\end{equation*}

Certainly, we can derive the arithmetic representation of $f$ by setting $\epsilon$ in \bb{chain rule} to zero\footnote{$f(\epsilon,\lambda) = n(\lambda+1)H\left(\frac{1}{\lambda+1}\right)-n(\epsilon\lambda+1)H\left(\frac{1}{\epsilon\lambda+1}\right).$}. However, we choose to retain the recursive formula because it offers valuable insights: \textbf{According to the chain rule, any membership problem  can be losslessly factorized into two (or multiple) stages.} The first stage involves a coarse membership algorithm that may yield some false positive items, while the second stage employs an accurate membership algorithm to efficiently handle the remaining small percentage of false positive items, especially when the negative-positive ratio is very large. The chain rule reveals a non-trivial observation that the factorization process incurs zero information loss. This observation serves as inspiration for algorithm design and provides a solid theoretical foundation for future research in the field of  membership problems.

Meanwhile, we acknowledge that the chain rule has certain limitations. First, it does not fully support dynamic memberships since it necessitates the identification of all false positive items before constructing the second stage. Moreover, we discover that the lossless property no longer holds for dynamic memberships (\bb{Section \ref{algo::replace elementary filters}}). Second, our chain rule theory is based on the assumption that all positive items are randomly selected from the universe, which may not perfectly accommodate real-world scenarios where item keys follow specific data distributions (\bb{Section \ref{Learned Filter}}).

\noindent$\bullet$ \bb{Paper Organization and Key Contributions.}

\textbf{(A, \S\ref{algo::membership query})} We propose a space lower bound for general membership problems and derive an elegant factorization theorem called chain rule. This theorem allows us to divide any membership problem into sub-problems without losing information.

\bb{(B, \textbf{\S\ref{algorithm}})} Building upon the chain rule, we introduce the versatile \algoname{} framework, which enables the combination of multi-stage membership algorithms, such as Bloomier Filters, to construct more efficient algorithms. The framework is compatible with assorted elementary filters, supports different combining operators (e.g. \texttt{AND} and \texttt{NAND}), and can handle certain dynamic scenarios.

\bb{(C, \S\ref{perf})} We evaluate the performance of \algoname{} in data compressing, classifying and filtering applications. Experimental results show that although our utilization of chain rule in this paper is simple and rudimentary, the ChainedFilter can significantly outperform existing works: (1) it requires only 26\% additional space over the theoretical lower bound for implicit static dictionary, (2) it requires only 0.22 additional bit per item over the theoretical lower bound for lossless data compression, (3) it reduces up to 31\% external memory access than raw Cuckoo Hashing \cite{fan2014cuckoo}, (4) it reduces up to 36\% \texttt{P99} tail point query latency than Bloom Filter under the same space cost in RocksDB database \cite{dong2021rocksdb}, and (5) it reduces up to 99.1\% filter space than original Learned Bloom Filter \cite{kraska2018case,mitzenmacher2018model,liu2020stable,dai2020adaptive}. We release our open source code on GitHub \cite{opensourcecode}.
\begin{equation*}
\begin{split}
\underbrace{\textbf{Chain rule}}_{\textbf{(A, \S\ref{algo::membership query}) Theory}}\xrightarrow[]{\textbf{design}}
\underbrace{\textbf{\algoname{}}}_{\textbf{(B, \textbf{\S\ref{algorithm}}) Algorithm}} \xrightarrow[]{\textbf{deploy}} \underbrace{\textbf{Applications}}_{\textbf{(C, \S\ref{perf}) Implementation}}
\end{split}
\end{equation*}

	\presub
\section{Chain Rule Theory}\label{algo::membership query}
\postsub

%
\renewcommand\arraystretch{1}
\begin{table}[h!tbp]
	\centering
	\begin{tabular}{c|l}
		\bottomrule
            \textbf{Symbol} & \textbf{Description} \\
            \hline
            $\m{U}$ & Universe (positive and negative items)\\
            $\m{S}$ & The set of positive items\\
            $\mathscr{S}$ & The set of all possible sets of positive items ($\m{S}$)\\
            $n$ & The number of positive items ($|\m{S}|$)\\
            $\epsilon$ & False positive rate\\
            $\lambda$ & Negative-positive ratio ($|\m{U}\backslash \m{S}|/|\m{S}|$)\\
            $x$ & An item\\
            $C$ & Any constant greater than 1\\
            $M$ & The number of buckets in one hash table\\
            $r$ & $|\m{U}|/(2M)$\\
            $N$ & The number of SSTables in one level of LSM-Tree\\
            $\sigma$ & A mapping from $\mathscr{S}$ to $\mathscr{F}$\\ 
        $\mathcal{X}_{\mathcal{F}}$& $\bigcup_{\mathcal{S}\in \sigma^{-1}(\mathcal{F})} \mathcal{S}$\\
            $\m{F}(\cdot)$ & A specific filter\\
            $\mathscr{F}$ & The set of all possible filters ($\m{F}$)\\
            $nf(\cdot,\cdot)$ & Space lower bound for membership problems\\
            $nf^{\m{F}}(\cdot,\cdot)$& Space cost of the membership filter $\m{F}(\cdot)$\\
            $H(\cdot)$ & Shannon's entropy\\
            $h_\alpha(\cdot)$& $\alpha$-bit hash value of an item\\
            $f_\alpha(\cdot)$& $\alpha$-bit fingerprint of an item\\
            $T[\cdot]$ & Hash table\\
            BF/CF/EF & Bloomier Filter, \algoname{}, elementary filter\\
            \toprule
	\end{tabular}
\caption{Commonly used symbols.}
\label{commonly used symbols}
\end{table}
\subsection{Definition and Notations}
Given a universe $\mathcal{U}$ and arbitrary subset $\mathcal{S}\subset \mathcal{U}$, where $|\mathcal{U}|$ and $|\mathcal{S}|$ are known, the membership problem is to determine whether a queried item $x\in \mathcal{U}$ is in $\mathcal{S}$. 

In this paper, we call a membership algorithm as a \bb{filter}. This definition is broader than what people commonly use\footnote{E.g., in prior arts, people often call an exact membership algorithm as a ``dictionary''.}. A filter is an indicator function $\mathcal{F}(\cdot): \mathcal{U}\mapsto \{0,1\}$  with one-sided error, which means $\forall x \in \mathcal{S}$ (positive item), $\mathcal{F}(x) = 1$ has zero false negative; for $x\in \mathcal{U}\backslash \mathcal{S}$ (negative item), we allow a small \bb{false positive rate} $\epsilon\in[0,1]$ s.t. $\mathcal{F}(x) = 1$. Moreover, we define $\lambda:=|\m{U}\backslash \m{S}|/|\m{S}|$ as the \bb{negative-positive ratio}, and we denote the membership problem as $(\epsilon,\lambda)$. This definition encompasses both approximate ($\epsilon\neq 0$ and $\lambda \to +\infty$) and exact ($\epsilon=0$ and $\lambda< +\infty$) membership problems as extreme cases. Unless specified otherwise, we assume that a filter is designed to handle static membership problems. If a membership filter also supports dynamic insertions of new items, we explicitly refer to it as a "dynamic filter."
%
For quick reference, We list all commonly used symbols in \bb{Table \ref{commonly used symbols}}\footnote{Elementary filter (the last row of \bb{Table \ref{commonly used symbols}}): When we combine several sub-filters to form a larger one, each sub-filter is called an elementary filter.}. We will define some of these symbols in later sections.

%
\subsection{Space Lower Bound}\label{Space Lower Bound}
In this part, we present a unified and complete space lower bound for general membership problems.

\begin{theorem}
\label{algo::lower bound} (\textbf{{Space Lower Bound}}) Ignore $o(n)$ terms. Let $nf{(\epsilon,\lambda)}$ $:=n\inf\limits_\mathcal{F}\mathcal{F}({\epsilon,\lambda})$ be the space lower bound for general membership problem $(\epsilon,\lambda)$, we have 
$$f(\epsilon,\lambda) = (\lambda+1)H\left(\frac{1}{\lambda+1}\right)-(\epsilon\lambda+1)H\left(\frac{1}{\epsilon\lambda+1}\right),$$
where $H(p):= -p\log p-(1-p)\log (1-p)\text{ is Shannon's entropy}.$
\end{theorem}

\begin{proof} The key idea of this proof follows the technique of \cite{carter1978exact}, except we take both $\epsilon$ and $\lambda$ into account. Given a membership problem ($\epsilon,\lambda$), we analysis all possible mappings from input sets to all possible filters, and use information theory to derive a lower bound of the space cost.

Given the universe $\mathcal{U}$ and the negative-positive ratio $\lambda$, we start by defining the set $\mathscr{S}:=\{\mathcal{S}\subset \mathcal{U} : |\mathcal{S}|=n \}$ as the set of all possible sets of positive items, and $\mathscr{F}$ as the set of all filter instances. Given any $\mathcal{S}\in \mathscr{S}$, we draw a filter instance $\mathcal{F} \in \mathscr{F}$ which has zero false negative. We denote this drawing method as a mapping $f: \mathscr{S} \mapsto \mathscr{F}$. Note that more than one set in $\mathscr{S}$ may map to a same filter $\mathcal{F}$, so the inverse mapping $\sigma^{-1}(\mathcal{F})\subset \mathscr{S}$ may have more than one element. Actually, we have
$$\sum\limits_{\mathcal{F}\in \mathscr{F}}|\sigma^{-1}(\mathcal{F})|=|\mathscr{S}|=\binom{|\mathcal{U}|}{|\mathcal{U}\backslash \mathcal{S}|} \text{ and } |\mathscr{F}|\leqslant 2^{n f(\epsilon,\lambda)}.$$

Since the filter $\mathcal{F}$ has zero false negative, so for every $x\in \mathcal{X}_{\mathcal{F}}:=\bigcup\limits_{\mathcal{S}\in \sigma^{-1}(\mathcal{F})} \mathcal{S}$, we have $\mathcal{F}(x)=1$ (\textbf{Figure \ref{pic::proof}}).

\begin{figure}[h!tbp]
  \centering
  \includegraphics[width=0.48\textwidth]{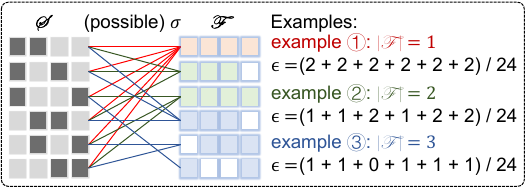}
\caption{Some examples when $\m{U}=\{1,2,3,4\}$ and $|\m{S}|=2$ (i.e.$\mathscr{S}=\{\{1,2\},\{1,3\},\{1,4\},\{2,3\},\{2,4\},\{3,4\}\})$. All connections are possible (but maybe not unique) mappings. The mapping should ensure that the filter has a zero false negative rate. For instance, in example 2, $\sigma$ maps the input positive set $\m{S}=\{1,2\}$ (the first row on the left) to the first green filter $\m{F}$ s.t. $\m{X}_\m{F}=\{1,2,3\}$. In other words, $\m{F}(1)=1,\m{F}(2)=1,\m{F}(3)=1$ (false positive), and $\m{F}(4)=0$.}
\label{pic::proof}
\end{figure}

The next two formulae are the critical observations in our proof: On the one hand, since the elements in $\sigma^{-1}(\mathcal{F})$ are different from each other, we have
$$\binom{|\mathcal{X}_{\mathcal{F}}|}{|\mathcal{S}|}\geqslant |\sigma^{-1}(\mathcal{F})|;$$

On the other hand, for a specific $\mathcal{S}\in \sigma^{-1}(\mathcal{F})$ and for all $x\in \mathcal{U}$, $\mathcal{F}(x)=1$ incurs false positive iff $x\in X_{\mathcal{F}}\backslash\mathcal{S}$. Therefore, the overall false positive rate is
$$\epsilon = \frac{\sum\limits_{\mathcal{F}\in \mathscr{F}} |\sigma^{-1}(\mathcal{F})| \frac{|\mathcal{X}_{\mathcal{F}}\backslash \mathcal{S}|}{|\mathcal{U}\backslash \mathcal{S}|}}{\sum\limits_{\mathcal{F}\in \mathscr{F}}|\sigma^{-1}(\mathcal{F})|}.$$

Now let $\theta>0$, our main idea of the following proof is to divide $\mathscr{F}$ into two parts $\mathscr{F}_1:=\{\mathcal{F}\in \mathscr{F}: |X_{\mathcal{F}}\backslash \mathcal{S}|\geqslant \theta n\}$ and $\mathscr{F}_2:=\mathscr{F}\backslash \mathscr{F}_1$, then bound their false positive rate respectively.

Consider the monotonically increasing function $g(p):=\exp\{n((p+1)\ln(p+1)-p\ln p)\} (p>0)$. To start with, we prove $\forall \varepsilon \in (0,1), \exists N_0>0, s.t. \forall \mathcal{F}\in\mathscr{F}_1,$ when $n>N_0,\varepsilon g^{-1}(|\sigma^{-1}(\mathcal{F})|)n$ is uniformly less than $|X_{\mathcal{F}}\backslash \mathcal{S}|.$ Initially, given $\varepsilon<\varepsilon'<1$, we can select $N_1=(\varepsilon'-\varepsilon)^{-1}$ s.t. when $n>N_1,$ we have $ \varepsilon g^{-1}(|\sigma^{-1}(\mathcal{F})|) n< \varepsilon' g^{-1}(|\sigma^{-1}(\mathcal{F})|)n$ $-1< \lfloor \varepsilon' g^{-1}(|\sigma^{-1}(\mathcal{F})|) n\rfloor.$ Next, according to Stirling's formula 
$$\lim\limits_{n\to+\infty}\frac{n!}{\sqrt{2\pi n}}\left(\frac{e}{n}\right)^n=1,$$
we can select $N_0>N_1$ which only relies on $\varepsilon$ and $\theta$, s.t. $\forall n>N_0$, $\tbinom{\lfloor\varepsilon' g^{-1}(|\sigma^{-1}(\mathcal{F})|) n\rfloor+|\mathcal{S}|}{|\mathcal{S}|}<2g(\varepsilon'g^{-1}(|\sigma^{-1}(\mathcal{F})|))<g(g^{-1}(|\sigma^{-1}(\mathcal{F})|))$ $=|\sigma^{-1}(\mathcal{F})|\leqslant \tbinom{|\mathcal{X}_{\mathcal{F}}|}{|\mathcal{S}|}\Rightarrow\varepsilon g^{-1}(|\sigma^{-1}(\mathcal{F})|) n< 
 \lfloor\varepsilon' g^{-1}(|\sigma^{-1}(\mathcal{F})|) n\rfloor< |X_{\mathcal{F}}\backslash \mathcal{S}| $ holds for all $\mathcal{F}\in \mathscr{F}_1$.

Similarly, when $n\to +\infty,$ we have $|\sigma^{-1}(\mathcal{F})|<g(2\theta)$ for all $\mathcal{F}\in \mathscr{F}_2$. Therefore 
$$ \sum\limits_{\mathcal{F}\in \mathscr{F}_1}|\sigma^{-1}(\mathcal{F})| \geqslant |\mathscr{S}|-  g(2\theta) |\mathscr{F}|.$$ 
To take a step further, consider $h(p):=pg^{-1}(p) := xp.$ Because 
$$\frac{\mathrm{d}h(p)}{\mathrm{d}p} = y + p/\left(\frac{\mathrm{d}g^{-1}(y)}{\mathrm{d}y}\right)=y+\frac{1}{n\ln (y+1) - n\ln y}$$ 
increases as $y$ increases (so as $p = g^{-1}(y)$ increases), we find $h$ convex. According to Jensen's inequality, we have
\begin{equation*}
\begin{split}
\epsilon &\geqslant \frac{\sum\limits_{\mathcal{F}\in \mathscr{F}_1} |\sigma^{-1}(\mathcal{F})| \frac{|\mathcal{X}_{\mathcal{F}}\backslash \mathcal{S}|}{|\mathcal{U}\backslash \mathcal{S}|}}{\sum\limits_{\mathcal{F}\in \mathscr{F}}|\sigma^{-1}(\mathcal{F})|}\geqslant \frac{\sum\limits_{\mathcal{F}\in \mathscr{F}_1} |\sigma^{-1}(\mathcal{F})| \varepsilon g^{-1}(|\sigma^{-1}(\mathcal{F})|)n}{|\mathscr{S}|\lambda n}\\
&\geqslant \frac{\varepsilon}{|\mathscr{S}| \lambda} \left(\sum\limits_{\mathcal{F}\in \mathscr{F}_1} |\sigma^{-1}(\mathcal{F})|\right) g^{-1}\left(\frac{\sum\limits_{\mathcal{F}\in \mathscr{F}_1} |\sigma^{-1}(\mathcal{F})|}{|\mathscr{F}_1|}\right) 
\end{split}
\end{equation*}
\begin{equation*}
\begin{split}
&\geqslant \frac{\varepsilon}{ \lambda}\left(1-\frac{g(2\theta)|\mathscr{F}|}{|\mathscr{S}|}\right)g^{-1}\left(\frac{|\mathscr{S}|}{|\mathscr{F}|}- g(2\theta)\right).
\end{split}
\end{equation*}
Let $\varepsilon\to 1, \theta\to 0$ and ignore $o(n)$, we have 
$$\epsilon \lambda \geqslant g^{-1}\left(\frac{|\mathscr{S}|}{|\mathscr{F}|}\right)\Rightarrow g(\epsilon \lambda)\geqslant \frac{|\mathscr{S}|}{|\mathscr{F}|} \geqslant \frac{g(\lambda)}{  2^{n f({\epsilon,\lambda})} }.$$
Therefore
$$f({\epsilon,\lambda}) \geqslant \frac{1}{n}\log \frac{g(\lambda)}{g(\epsilon \lambda)}=\log\left(\frac{(\lambda+1)^{\lambda+1}}{\lambda^\lambda}\right)-\log\left(\frac{(\epsilon \lambda+1)^{\epsilon \lambda+1}}{(\epsilon \lambda)^{\epsilon \lambda}}\right).$$
It's easy to verify that $f({\epsilon,\lambda})\leqslant\frac{1+o(1)}{n}(\log g(\lambda)- \log g(\epsilon \lambda))$ (consider the example $\binom{|\mathcal{X}_{\mathcal{F}}|}{|\mathcal{S}|} =(|\mathscr{S}|/|\mathscr{F}|)(1+o(1))$). So 
$$f({\epsilon,\lambda})=(\lambda+1)H\left(\frac{1}{\lambda+1}\right)-(\epsilon\lambda+1)H\left(\frac{1}{\epsilon\lambda+1}\right)$$
is the space lower bound.
\end{proof}
\begin{Rmk}\textbf{Theorem \ref{algo::lower bound}} connects the approximate and the exact membership query problems.
%
When $\epsilon\neq 0$ and $\lambda\to+\infty$,
$f({\epsilon,+\infty})=\log1/\epsilon$
degenerates to the space lower bound of approximate memberships. When $\epsilon = 0$ and $\lambda<+\infty$,
$f({0,\lambda})=(\lambda+1)H\left(1/(\lambda+1)\right)$
degenerates to the space lower bound of exact memberships \cite{carter1978exact}.
\end{Rmk}
\subsection{Chain Rule}
In this part, we delve deeper and derive our chain rule theory, which is arithmetically equivalent to the space lower bound but provides insights into the essence of membership problems. As a preliminary, please note that in our context, the term ``membership problem'' is different from the term ``membership filter''. When referring to an abstract ``problem'', we focus on the theoretical space lower bound. While when discussing a specific ``filter'', we focus on a practical algorithm that may not be space-optimal.

To begin, let's consider encoding $n$ positive items and $\lambda n$ negative items with a false positive rate of $\epsilon_1\epsilon_2$. Intuitively, we can factorize the problem into two stages: first, encoding all positive items and the $\lambda n$ negative items with a false positive rate of $\epsilon_1$, and then encoding the positive items and the remaining $ \lambda \epsilon_1 n$ false positive items with a false positive rate of $\epsilon_2$. It is evident that if we solve the two sub-problems, we can address the primary problem. Hence, the two sub-problems should not be easier than the primary problem. However, since the positive items are encoded twice in separate stages, it may seem that the two-stage factorization incurs additional space overhead, and we might need to make careful trade-offs to avoid accumulating inherent space costs caused by the factorization. 

Surprisingly, \bb{Theorem \ref{algo::lower bound}} reveals that all our concerns and worries are unnecessary. Because
\begin{equation*}
\begin{split}
    f(\epsilon_1\epsilon_2,\lambda)&=(\lambda+1)H\left(\frac{1}{\lambda+1}\right)-(\epsilon_1\epsilon_2\lambda+1)H\left(\frac{1}{\epsilon_1\epsilon_2\lambda+1}\right)\\
    &=\left((\lambda+1)H\left(\frac{1}{\lambda+1}\right)-(\epsilon_1\lambda+1)H\left(\frac{1}{\epsilon_1\lambda+1}\right)\right)\\
    &+\left((\epsilon_1\lambda+1)H\left(\frac{1}{\epsilon_1\lambda+1}\right)-(\epsilon_1\epsilon_2\lambda+1)H\left(\frac{1}{\epsilon_1\epsilon_2\lambda+1}\right)\right)\\
    &=f(\epsilon_1,\lambda)+f(\epsilon_2,\epsilon_1\lambda),
\end{split}
\end{equation*}

we find the factorization is completely lossless (ignoring $o(1)$ terms), which means we can arbitrarily decompose any membership problem into an arbitrary number of sub-problems without incurring any additional space cost. In other words,  if all the elementary filters used in the combination are space-optimal, the resulting combined membership filter will also be space-optimal\footnote{As an extreme example, when $\epsilon_2=1$, the equation $f(\epsilon_1,\lambda)=f(\epsilon_1,\lambda)+f(1,\epsilon_1\lambda)$ also holds because the second stage filter can always report true and thus does not contribute to the space cost.}. Instead, when we combine imperfect membership filters to solve a membership problem, the only source of space overhead stems from that the filters cannot optimally solve the sub-problems. The factorization process itself does not introduce any additional space overhead. In later chapters, we will use Bloomier Filters as a straightforward example to show how to appropriately factorize a membership problem into sub-problems and enhance overall performance. But before that, let us rewrite the expressions as $\epsilon=\epsilon_1\epsilon_2$ and $\epsilon'=\epsilon_1$, and present the conclusion in a recursive form: 
\begin{theorem}\label{thm::chain rule}
    (\textbf{Chain Rule Theory}) \begin{equation*}
    \begin{cases}
    f\left(\epsilon, \lambda\right) = {f\left(\epsilon',\lambda\right)}+{f\left(\epsilon/{\epsilon'}, \epsilon'\lambda\right)},\forall \epsilon'\in[\epsilon,1];\\
        f\left(0, \lambda\right)=\left(\lambda+1\right)H\left(\frac{1}{\lambda+1}\right).
    \end{cases}
\end{equation*}
\end{theorem}

It's easy to verify that the above equation set is equivalent to the space lower bound in \textbf{Theorem \ref{algo::lower bound}} (you can check it by setting $\epsilon=0$), but this form is more elegant and may provide more insights.
\presub
\section{Bloomier Filter}\label{algo::warm up}
\postsub

In this section, we present two variants derived from the Bloomier Filter \cite{chazelle2004bloomier,charles2008bloomier}, which belong to approximate and exact membership algorithms, respectively\footnote{The approximate Bloomier Filter is also referred to as XOR filter \cite{graf2020xor} or binary fuse filter \cite{graf2022binary}. For consistency, we use the term ``Bloomier Filter'' interchangeably throughout this paper.}. In \bb{Section \ref{algorithm}}, we combine these elementary filters by chain rule to construct \algoname{}. The reason we introduce Bloomier Filters as elementary filters is that they are easy to describe and implement, but readers can also use other elementary filters to achieve special properties, like smaller filter space, smaller construction space or supporting dynamic exclusions.
We discuss more related works in \textbf{Section \ref{sec:related}}.

\textbf{Overview.} Bloomier Filter is a compact perfect hashing algorithm that supports both approximate and exact membership query by encoding item $e$'s fingerprint $f_\alpha (e)\in \{0,1\}^\alpha$ into the hash table. Specifically, to build an approximate membership algorithm, the Bloomier Filter encodes an $\alpha$-bit fingerprint $f_\alpha(e)=h_\alpha(e)$ for every positive item; to build an exact membership algorithm, the Bloomier Filter maps every item to a one-bit hash value $h_1(e)\in\{0,1\}$ and encodes $f_1(e)=h_1(e)$ (resp. $f_1(e)=\sim h_1(e)$) for every positive (resp. negative) item. The query result of an item $e$ depends on whether its hash value matches the fingerprint in the hash table. For interested readers, we present the detailed descriptions of Bloomier Filter in the next three paragraphs. \bb{Skipping them and directly reading the \textbf{Remark} does not affect the comprehension of this paper}.

\textbf{Algorithm.} Suppose we already have the entire universe of all $|\m{U}|$ possible items\footnote{Bloomier Filter is a static algorithm and does not support dynamic scenarios.} whose value are either zero or one. To construct the hash table, we first initialize the hash table by all zero, mark all items as ``not matched'' and let the variable $order=|\m{U}|$. the Bloomier Filter maps every not-matched item $e$ to $j$ different slots $s^e_{[1..j]}$. Then, it repeats the following operations (called the \textit{peeling} process) until all items are ``matched'': (1) It selects a slot $s$ that is mapped by only one item $e_0$\footnote{If the algorithm cannot find such a slot, it reports construction fail and terminates. But the theory later proves that the algorithm will succeed with high probability.}; (2) It marks the insertion place (ip) of $e_0$ as $s(e_0)=s$, and the insertion order (io) of $e_0$ as $order$; (3) It marks $e_0$ as ``matched'', peels $e_0$ and decreases $order$ by one. Finally, it inserts items in $order$: for items $e$ from $order = 1$ to $|\m{U}|$, it encodes $(\oplus_{i=1}^j s^{e}_{i})\oplus f_\alpha(e)$(where $\oplus$ means XOR) into the slot $s(e)$. To query an item $e$, the Bloomier Filter reports $\oplus_{i=1}^j s^{e}_{i}$ as the result.

\textbf{Example.} Suppose we have $|\m{U}|=3$ items $e_1,e_2$ and $e_3$. To construct the hash table, the Bloomier Filter first maps $e_1$ to $i=3$ slots $s_1,s_2,s_3$; maps $e_2$ to $s_1,s_3,s_4$; and maps $e_3$ to $s_2,s_4,s_5$. Then the filter (1) selects slot $s_5$ which is only mapped by $e_3$, marks $e_3$ with (ip = 5, io = 3) and peels $e_3$; (2) selects $s_2$ which is only mapped by $e_1$, marks $e_1$ with (ip = 2, io = 2) and peels $e_1$; (3) selects $s_1$ which is only mapped by $e_2$, marks $e_2$ with (ip = 1, io = 1) and peels $e_2$. Finally, the filter (1) inserts $s_1$ with $s_1\oplus s_3\oplus s_4 \oplus f_\alpha(e_2)=f_\alpha(e_2)$ (io = 1); (2) inserts $s_2$ with $s_1\oplus s_2\oplus s_3 \oplus f_\alpha(e_1)=f_\alpha(e_2)\oplus f_\alpha(e_1)$ (io = 2); (3) inserts $s_5$ with $s_2\oplus s_4\oplus s_5 \oplus f_\alpha(e_3)=f_\alpha(e_2)\oplus f_\alpha(e_1) \oplus f_\alpha(e_3)$  (io = 3). To query the items, the Bloomier Filter reports $Query(e_1)=s_1\oplus s_2\oplus s_3=(f_\alpha(e_2))\oplus(f_\alpha(e_2)\oplus f_\alpha(e_1))\oplus 0 =f_\alpha(e_1), Query(e_2)= s_1\oplus s_3\oplus s_4=(f_\alpha(e_2))\oplus 0 \oplus 0 = f_\alpha(e_2), \text{ and } Query (e_3)=s_2\oplus s_4 \oplus s_5 = (f_\alpha(e_2)\oplus f_\alpha(e_1))\oplus 0\oplus (f_\alpha(e_2)\oplus f_\alpha(e_1)\oplus f_\alpha(e_3))=f_\alpha(e_3)$.

\textbf{Theory.} The theory underlying the Bloomier Filter is highly non-trivial \cite{goodrich2011invertible}. The peeling process mentioned above is actually equivalent to finding the 2-core of random hypergraph. If the number of potential slots is greater than $c N$, where $c$ is any constant larger than 
$$c_j^{-1}:=\left(\sup \left\{\alpha\in(0,1): \forall x\in(0,1), 1-e^{-j\alpha x^{j-1}}<x\right\}\right)^{-1},$$
the peeling process will succeed with high probability $1-o(1)$. We denote $c_j^{-1}$ as the \textit{threshold} which achieves its minimum value $1.23$ when $j=3$. Recently, this result is optimized by \cite{walzer2021peeling} (2021) which proposes a new distribution of hyperedges. The author selects $j$ slots uniformly at random from a range of $|\m{U}|/(z+1)$ slots, improving the threshold close to $c_j'^{-1}$, where $c_j'^{-1}$ can be arbitrarily close to 1 (e.g. $c_3'^{-1} \approx 1.12, c_4'^{-1} \approx 1.05$ when $z=120$). In \textbf{Section \ref{perf}}, we set $j=3,z=120,$ and $C=1.13$ to conduct experiments.
\begin{Rmk}
To conclude, the Bloomier Filter can derive space-efficient approximate ($Cn\log 1/\epsilon$ bits) and exact ($C|\mathcal{U}|$ bits) membership filters, where $C$ can be an arbitrary constant close to 1. Interestingly, the two formulae have completely different forms, which seemingly leave room for improvement. This simple observation inspires us to bridge the gap between the two variants to obtain a better theoretical result.
\end{Rmk}
        \section{\algoname{}}\label{algorithm}

\subsection{Key Insight}\label{algo::key design}
We use the exact membership \algoname{} as an example to illustrate our key insight. Now let's recap the theoretical results of the approximate and the exact variants of Bloomier Filters (BF) shown in the \textbf{Remark} of \textbf{Section \ref{algo::warm up}}:
\begin{equation*}
\begin{cases}
f^{BF}({\epsilon, \lambda }) \leqslant f^{BF}(\epsilon,+\infty)= C\log 1/\epsilon &\text{(Approximate)};\\
f^{BF}({0, \lambda }) = C(\lambda +1) &\text{(Exact)}.
\end{cases}
\end{equation*}

According to \bb{Theorem \ref{thm::chain rule}}, we can losslessly factorize a membership problem ($0, \lambda$) into two sub-problems ($\epsilon',\lambda$) and ($0,\epsilon'\lambda$), and use Bloomier Filters as elementary filters to form \algoname{} s.t.
$$f^{CF}(0,\lambda)=f^{BF}(\epsilon',\lambda)+f^{BF}(0,\epsilon'\lambda)\leqslant C\log 1/\epsilon'+C(\epsilon'\lambda+1).$$
Because $\epsilon'$ can be any value between 0 and 1, we can minimize $f^{CF}(0,\lambda)$ to $C\log (2e\lambda \ln 2)$ by setting\footnote{We assume $\frac{1}{ \lambda \ln 2}<1$, otherwise it degenerates to the exact Bloomier Filter. For convenience, we \textbf{don't round} numbers here, but we'll do it in practice.} $\epsilon'=1/(\lambda\ln 2).$
Such an intuitive technique amazingly makes $f^{CF}(0,\lambda)$ less than $1.11f({0,\lambda})$. 
\begin{algorithm}
\renewcommand\baselinestretch{0.9}\selectfont
\LinesNumbered
  \caption{\algoname{} for Exact Membership Query}\label{algo::construct}
  \KwIn{Universe $\mathcal{U}$ and subset $\mathcal{S}$, $|\mathcal{U}|/|\mathcal{S}|= \lambda >1/\ln 2$.}  
  \KwOut{A filter $\mathcal{F}:\mathcal{U}\mapsto \{0,1\}$ s.t. $\mathcal{F}(e)=1$ iff $e\in S$.}
  \textbf{Function} Construct ($\mathcal{U},\mathcal{S}$):\\
  Set $\log 1/\epsilon=\lfloor \log \lambda\rfloor$ and the false positive set $\mathcal{S}' = \emptyset.$\\
  Construct an approximate membership filter $\mathcal{F}_1 $ s.t.  
  \begin{equation*}
      \begin{cases}
  \mathcal{F}_1(e)=1, & \forall e\in \mathcal{S};\\
  \mathbb{P}[\mathcal{F}_1(e)=0] = \epsilon, & \forall e\in \mathcal{U}\backslash\mathcal{S}.
  \end{cases}\end{equation*}
  \textbf{For all} $e\in\mathcal{U}\backslash\mathcal{S}$ \textbf{satisfying} $\mathcal{F}_1(e)=1:$ Insert $e$ into $\mathcal{S}'.$\\
  Construct an exact filter $\mathcal{F}_2 $ s.t.  \begin{equation*}
\begin{cases}
  \mathcal{F}_2(e)=1, & \forall e\in \mathcal{S};\\
  \mathcal{F}_2(e)=0, & \forall e\in \mathcal{S}'.\\
  \end{cases}      
  \end{equation*}
  \textbf{return} $\mathcal{F}(\cdot):=\mathcal{F}_1(\cdot)$ \& $\mathcal{F}_2(\cdot)$.
\end{algorithm}

This exact membership example demonstrates the potential of the chain rule. We summarize it in \textbf{Algorithm \ref{algo::construct}} and \textbf{Figure \ref{pic::algo1}}, and show its experimental performance in \textbf{Section \ref{perf::compact truth table}}. To construct \algoname{}, in step \textcircled{1}, we first encode all the positive items into a $Cn\log  (\lambda \ln 2)$-bit approximate Bloomier Filter with false positive rate $\epsilon=1/ (\lambda \ln 2)$ (Line 3). In step \textcircled{2}, we query all negative and collect the false positive ones into a set $\mathcal{S}'$ (Line 4). In step \textcircled{3}, we encode all the positive as well as all the false positive items (i.e. items in $\mathcal{S}'$) into a $Cn\log 2e$-bit Exact Bloomier Filter (Line 5). In this way, the overall space cost $nf^{CF}({0,  \lambda })$ is minimized. To query an item, we query both of the two filters and report their \texttt{AND} value as the result (Line 6).
\begin{figure}[h!tbp]
  \centering
   \includegraphics[width=0.45\textwidth]{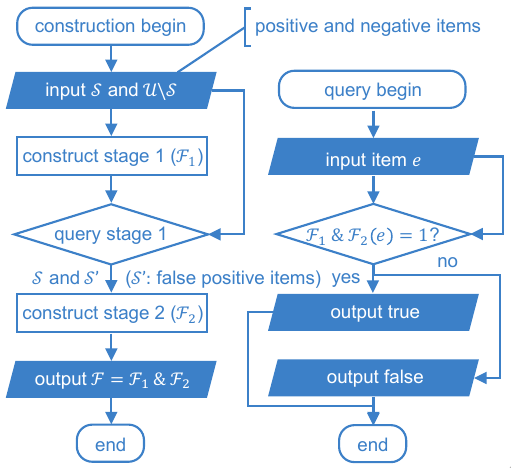}
\caption{\algoname{} shown in Section \ref{algo::key design}.}
\label{pic::algo1}
\end{figure}

\subsection{Generalization and Analysis}\label{algo::general}

In this part, we (1) extend the exact ($\epsilon=0$) \algoname{} to a general ($\epsilon\neq 0$) \algoname{} and show that the two-stage structure is space-optimal, (2) extend Bloomier Filters to an arbitrary number of elementary filters and reveal \algoname{}'s limitation. \bb{Since this part is only technically complex, skipping it does not affect the comprehension of this paper}. 

\subsubsection{\textbf{\algoname{} for general membership query}} 

Since the trivial theoretical bound $$f^{CF}({\epsilon, \lambda })\leqslant \min\left\{f^{CF}({0,(1-\epsilon) \lambda }),tf^{CF}({\epsilon, \lambda /t})\right\}(\forall t>1)$$ is too loose, we introduce additional inequalities to reduce error. The key observation is that an $C(\beta+1)n(\beta>0)$-bit Bloomier Filter can also filter out some negative items that are not encoded in the perfect hash table. Recall that in \textbf{Section \ref{algo::warm up}}, we say a Bloomier Filter encodes a one-bit fingerprint $f_1(e)= h_1(e)\in \{0,1\}$ (resp. $\sim h_1(e)$) for a positive (resp. negative) item. Now let's consider two strategies for generating $h_1(e)$. (a) The first strategy is $\mathbb{P}[h_\alpha(\cdot)=1]=1/2$, which means we flip a fair coin and record the face up side (resp. face down side) as the fingerprint of a positive (resp. negative) item. In this way, only 1/2 not-encoded negative items are false positive items. (b) The second strategy is $\mathbb{P}[h_\alpha(\cdot)=1]=1$, which means we directly record 1 (resp. 0) as the fingerprint of a positive (resp. negative) item. In this way, only $1/(\beta+1)$ not-encoded negative items are false positive items. These strategies give rise to two additional inequalities (BF represents \algoname{}. When the equation ``='' holds, the algorithm degenerates to one single Bloomier Filter)
\begin{equation*}
\begin{cases}
    f^{CF}({\epsilon, \lambda })/C\leqslant  \lambda +1-2\epsilon  \lambda \quad\quad\quad( \epsilon\leqslant2\lambda/(\lambda+1)) &(a); \\
    f^{CF}({\epsilon, \lambda })/C\leqslant  (\lambda +1)/(\epsilon  \lambda +1)\quad(\epsilon\leqslant 1/2    ) &(b).
\end{cases}
\end{equation*}
Similar to the exact \algoname{} (\textbf{Algorithm \ref{algo::construct}}), the generalized algorithm also consists of one approximate Bloomier Filter and one exact Bloomier Filter. The only difference is, the approximate Bloomier Filter requires $Cn\alpha=Cn(f^{CF}({\epsilon, \lambda })-\beta-1)$ bits and has a false positive rate of $1/2^\alpha$, while the Exact Bloomier Filter requires $Cn(\beta+1)$ bits and has a false positive rate of $\min\{1/2,1/(\beta+1)\}$. Using the two inequalities and the chain rule, we can determine the optimal parameter settings of $\alpha,\beta$ and thus $f^{CF}({\epsilon, \lambda })$. Because the calculation process is a bit dry, we present the results directly in \textbf{Corollary \ref{algo::Optimality}} and \textbf{Figure \ref{pic::space cost}}.
\begin{Cor}\label{algo::Optimality}
if we only combine two Bloomier Filters, the optimal space cost and the corresponding parameters are
$$f^{CF}({\epsilon, \lambda }):= \min\left\{f^{(a)}({\epsilon, \lambda }),f^{(b)}({\epsilon, \lambda })\right\},$$ 
where $f^{(a)}({\epsilon, \lambda }),f^{(b)}({\epsilon, \lambda })$ are defined as follows: 
\begin{equation*}
\begin{split}
    \text{(a)}&\text{ If } \lambda >\frac{1}{\ln2}\text{ and } \lambda <\frac{1}{2\epsilon\ln 2}\text{, then }\mathbb{P}[h_\alpha(\cdot)=1]=1/2,
    \end{split}
    \end{equation*}
    \begin{equation*}
\begin{split}
    &\beta_{(a)}=\frac{1}{\ln 2}-2 \lambda \epsilon\text{ and }f^{(a)}({\epsilon, \lambda })/C=\log (2e \lambda \ln 2)-2 \lambda \epsilon. \\
    &\text{ Otherwise $f^{(a)}({\epsilon, \lambda })$ degenerates to the space of approximate }\\
    & (\beta=0) \text{ or exact }(\alpha=0)\text{ Bloomier Filters.} \\
    \text{(b)}&\text{ If } \lambda >\frac{1}{\ln2-\epsilon}>0\text{, then }\mathbb{P}[h_\alpha(\cdot)=1]=1,\\
    &\beta_{(b)}=\frac{1}{\ln 2}-\frac{\epsilon  \lambda }{\epsilon  \lambda  + 1},\text{ and }f^{(b)}({\epsilon, \lambda })/C=\log \frac{2e \lambda \ln 2}{\epsilon  \lambda  + 1}-\frac{\epsilon  \lambda }{\epsilon  \lambda + 1}.  \\
    &\text{ Otherwise $f^{(b)}({\epsilon, \lambda })$ degenerates to the space of approximate }\\
    & (\beta=0)\text{ or exact }(\alpha=0)\text{ Bloomier Filters.} \\
\end{split}
\end{equation*}
\end{Cor}
\begin{figure}[h!tbp]
  \centering
   \includegraphics[width=0.48\textwidth]{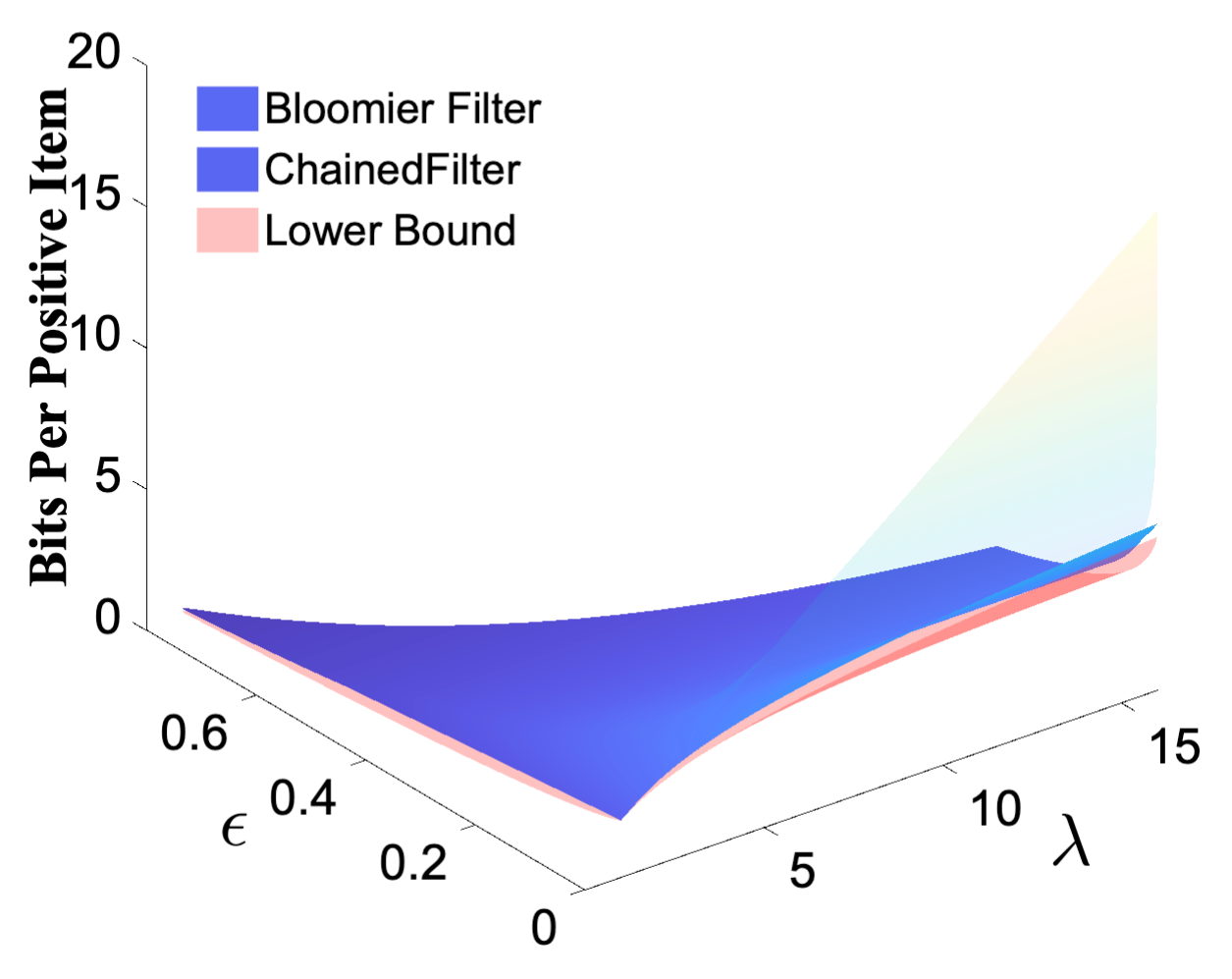}
\caption{Space cost when $C\to 1$. The multicolored surface is the minimum space cost of a single Bloomier Filter. When $\epsilon=0$ and $\lambda=16$, it is 210\% higher than the theoretical lower bound, while the space overhead of \algoname{} (the blue surface) is only 8\%.}
\label{pic::space cost}
\end{figure}

\subsubsection{\textbf{Two filters can be optimal}} It seems that only two stages is too trivial. But interestingly, we prove that combining two Bloomier Filters using operator ``\&'' is space-optimal.

\begin{theorem}\label{algo::optimality}
    (\textbf{Optimality}) Suppose we have a combined filter $\mathcal{F}:=\&_{i=1}^m\mathcal{F}_i$, where $\mathcal{F}_i$ are Bloomier Filters. Then we can prove that (CF represents \algoname{}) $$f^{\mathcal{F}}({\epsilon, \lambda })\geqslant f^{CF}({\epsilon, \lambda }).$$
\end{theorem}

\begin{proof}
We only consider inequality $(a)$, i.e. $f^{BF}({\epsilon, \lambda })/C= \lambda +1-2\epsilon  \lambda (\epsilon\leqslant2\lambda/(\lambda+1))$, where BF represents Bloomier Filter, because the case for inequality $(b)$ is similar.

Let $\epsilon:=\epsilon_1\epsilon_2...\epsilon_m$, according to the chain rule
\begin{equation*}
\begin{split}
f(\epsilon,\lambda)&=f(\epsilon_1,\lambda)+f(\frac{\epsilon}{\epsilon_1},\epsilon_1\lambda)\\
&=f(\epsilon_1,\lambda)+f(\epsilon_2,\epsilon_1\lambda)+f(\frac{\epsilon}{\epsilon_1\epsilon_2},\epsilon_1\epsilon_2\lambda)\\
&=f(\epsilon_1,\lambda)+f(\epsilon_2,\epsilon_1\lambda)+f(\epsilon_3,\epsilon_1\epsilon_2\lambda)+f(\frac{\epsilon}{\epsilon_1\epsilon_2\epsilon_3},\epsilon_1\epsilon_2\epsilon_3\lambda)\\
&=...\\
&=f(\epsilon_1,\lambda)+f(\epsilon_2,\epsilon_1\lambda)+...+f(\epsilon_m,\epsilon_1\epsilon_2...\epsilon_{m-1}\lambda),\\
\end{split}
\end{equation*}
we have
\begin{equation*}
\begin{split}
\min\limits_{\mathcal{F}}& f^{\mathcal{F}}({\epsilon, \lambda })/C= \min\limits_{\epsilon_1,\epsilon_2,...,\epsilon_m\in[0,1/2]}^{\epsilon_1 \epsilon_2 ... \epsilon_m=\epsilon}\sum\limits_{i=1}^{m}f^{BF}({\epsilon_i,\epsilon_1\epsilon_2...\epsilon_{i-1} \lambda })/C\\
&=m+(1-2\epsilon)  \lambda -\max\limits_{\epsilon_1,\epsilon_2,...,\epsilon_{m-1}\in[0,1/2]}\left(\sum\limits_{i=1}^{m-1}\prod\limits_{j=1}^{i}\epsilon_j\right) \lambda \\
&=m+(1-2\epsilon) \lambda -\left(1-\frac{1}{2^{m-1}}\right) \lambda=\left(m+\frac{\lambda}{2^{m-1}}\right)-2\epsilon\lambda.
\end{split}
\end{equation*}
\begin{equation*}
\text{When }\begin{cases}
    m=\lfloor\log  \lambda \rfloor +1; \\
    \epsilon_1=...=\epsilon_{m-1}=1/2;\\
    \epsilon_m=2^{m-1}\epsilon,
\end{cases}
\end{equation*}
The formula achieves the minimum value 
\begin{equation*}
f^{\mathcal{F}}({\epsilon, \lambda })=\lfloor\log  \lambda \rfloor+1+\frac{ \lambda }{2^{\lfloor\log  \lambda \rfloor}}-2\epsilon  \lambda=f^{CF}({\epsilon, \lambda }).
\end{equation*}

So \algoname{} is space-optimal.
\end{proof}
\begin{Rmk}When $\epsilon=0$, the rounded space cost of \algoname{} is 
$$f^{CF}({0, \lambda })=C(\lfloor\log  \lambda \rfloor+1+\frac{ \lambda }{2^{\lfloor\log  \lambda \rfloor}})< 1.11Cf({0,  \lambda }).$$
 We will recap this result in \textbf{Section \ref{perf::compact truth table} and \ref{perf::randomized huffman coding}}.
\end{Rmk}

\subsubsection{\textbf{Limitation of the operator ``\&''}} 
When the elementary filters are not limited to Bloomier Filters, the optimal combined filter can be very complex. In fact, we may not even know how many elementary filters we should use. However, we find that if we continue to use the ``\&'' operator to combine filters, we can derive the tight space lower bound of the combined filter.

\begin{theorem}\label{algo::limit}
(\textbf{Limitation}) 
Given arbitrary elementary filters (EF) $\mathcal{F}_i$ with a restriction of $f^{EF}({\epsilon, \lambda })$ and an arbitrary combined filter $\mathcal{F}:=\&_{i=1}^m\mathcal{F}_i$. We define $\Psi^{EF}({\epsilon, \lambda }):[0,1]\times \mathbb{R}^+\mapsto \mathbb{R}$ as any function satisfying
\begin{equation*}
\begin{cases}
f({\epsilon, \lambda })\leqslant \Psi^{EF}({\epsilon, \lambda })\leqslant f^{EF}({\epsilon, \lambda })&\forall \epsilon \in[0,1],\\
\Psi^{EF}({\epsilon_1\epsilon_2, \lambda })\leqslant \Psi^{EF}({\epsilon_1, \lambda })+\Psi^{EF}({\epsilon_2,\epsilon_1  \lambda })&\forall \epsilon_1,\epsilon_2\in[0,1].
\end{cases}
\end{equation*}
Then we have
\begin{equation*}
\inf f^{\mathcal{F}}({\epsilon, \lambda })=\sup \Psi^{EF}({\epsilon, \lambda }).
\end{equation*}
\end{theorem}
\begin{proof}

First, we assert that (1) $\inf f^{\mathcal{F}}({\epsilon, \lambda })$ exists. This is because  $f^{\mathcal{F}}({\epsilon, \lambda })$ exists ($f^{\mathcal{F}}({\epsilon, \lambda })\leqslant f^{EF}({\epsilon, \lambda })$) and has an lower bound $f({\epsilon, \lambda })$. (2) $\sup \Psi^{EF}({\epsilon, \lambda })$ exists. This is because  $\Psi^{EF}({\epsilon, \lambda })$ exists ($\Psi^{EF}({\epsilon, \lambda })\geqslant f({\epsilon, \lambda })$) and has an upper bound $f^{EF}({\epsilon, \lambda })$.

Second, we prove that $f^{\mathcal{F}}({\epsilon, \lambda })\geqslant \Psi^{EF}({\epsilon, \lambda })$ by mathematical induction. When $m=1$ (note that $m$ is the number of elementary filters), we observe $f^{\mathcal{F}}({\epsilon, \lambda })=f^{EF}({\epsilon, \lambda })\geqslant \Psi^{EF}({\epsilon, \lambda }).$ Assume that when $ m=m_0-1$ we have $f^{\mathcal{F}}({\epsilon, \lambda })\geqslant \Psi^{EF}({\epsilon, \lambda })$. Then, when $m=m_0,$ we can rewrite $\mathcal{F}$ as $(\&_{i=1}^{m_0-1}\mathcal{F}_i)\&\mathcal{F}_{m_0}.$ W.l.o.g., we can assume that $\&_{i=1}^{m_0-1}\mathcal{F}_i$ encodes $ \lambda n$ negative items with false positive rate $\epsilon'$, and $\mathcal{F}_{m_0}$ encodes $ \lambda 'n ( \lambda '\geqslant \epsilon' \lambda )$ negative items with false positive rate $\epsilon'' (\epsilon'\epsilon''\leqslant \epsilon)$. Thus 
\begin{equation*}
    \begin{split}
        f^{\mathcal{F}}({\epsilon, \lambda })&\geqslant\Psi^{EF}({\epsilon', \lambda })+\Psi^{EF}({\epsilon'', \lambda '})\\
        &\geqslant \Psi^{EF}({\epsilon', \lambda })+\Psi^{EF}({\epsilon'',\epsilon  \lambda })\\
        &\geqslant \Psi^{EF}({\epsilon'\epsilon'',  \lambda })\geqslant \Psi^{EF}({\epsilon,  \lambda }).
    \end{split}
\end{equation*}
Therefore, $\forall m\in \mathbb{N}^+$ we have $$f^{\mathcal{F}}({\epsilon, \lambda })\geqslant \Psi^{EF}({\epsilon, \lambda }) \Rightarrow \inf f^{\mathcal{F}}({\epsilon, \lambda })\geqslant \sup\Psi^{EF}({\epsilon, \lambda }).$$

Finally, we assert that 
\begin{equation*}
    \begin{split}
        \inf f^{\mathcal{F}}(&{\epsilon_1\epsilon_2, \lambda })\leqslant \inf f^{\mathcal{F}}({\epsilon_1, \lambda })+\inf  f^{\mathcal{F}}({\epsilon_2,\epsilon_1  \lambda })\\
        &\Rightarrow\inf f^{\mathcal{F}}({\epsilon, \lambda })\leqslant \sup \Psi^{EF}({\epsilon, \lambda }).
    \end{split}
\end{equation*}

Therefore, we conclude that $\inf f^{\mathcal{F}}({\epsilon, \lambda })=\sup \Psi^{EF}({\epsilon, \lambda}).$
\end{proof}

\begin{Rmk}\textbf{Theorem \ref{algo::limit}} shows the limitation of the operator ``\&'': if we solely use the ``\&'' operator, the space cost of a combined filter cannot be less than $\sup \Psi_{\epsilon, \lambda }^{EF}$. However, if we utilize the chain rule with other operators such as ``$\&\sim$'', then the situation may be different (\textbf{Section \ref{algo::dsbe}}).
\end{Rmk}

\subsection{\algoname{} is a Framework}\label{algo::framework}

In \textbf{Section \ref{algo::key design}}, we combine Bloomier Filters with operator ``\&'' to construct \algoname{}. However, this design has two limitations. First, Bloomier Filter only supports static membership query, which means that any new item added may require a reconstruction of the entire data structure. Second, although the Bloomier Filter itself only requires $O(n)$ space, its construction process requires an additional $\Omega(n\log n)$ space (as the peeling process relies on a good ordering). In this part, we present some extensions to \algoname{} that overcomes these shortcomings to some extent.

\subsubsection{Replace elementary filters}\label{algo::replace elementary filters}
\ 

For the first limitation, we can replace the Bloomier Filter(s) with other dynamic elementary filter(s) to support online updates. 
For instance, we can replace the static second stage (exact) Bloomier Filter with a dynamic filter such as Othello Hashing \cite{yu2018memory} or Coloring Embedder \cite{tong2019coloring}, at the expense of additional space\footnote{These two algorithms map each item to an edge in a random hypergraph and then color the nodes the same for positive items and different for negative items. Othello Hashing / Coloring Embedder requires 2.33/2.2 bits per item (compared to $C<1.13$ bits per item cost of Bloomier Filter) but support online updates.}. With the help of the dynamic ``whitelist'', the new version of \algoname{} supports the exclusion of new negative items without causing any false negative. 
Similarly, we can further replace the static first stage (approximate) Bloomier Filter, which requires $C\log 1/\epsilon$ bits per item, with a dynamic filter such as a Bloom Filter \cite{bloom1970space}, which requires $(\log 1/\epsilon)/\ln 2$ bits per item, or a Cuckoo Filter \cite{fan2014cuckoo}, which requires $1.05(2+\log 1/\epsilon)$ bits per item. In this way, the new version of \algoname{} supports not only the  exclusion of new negative items but also the inclusion (insertion) of new positive items with a small false positive rate. These enhancements enable \algoname{} to be used in more dynamic scenarios where new items are frequently added. 

While we present some compensations, it is important to note that \algoname{} does not perfectly align with dynamic scenarios. This is because we need to determine all false positive items before we construct the second stage filter. At the theoretical level, we can even prove that {the chain rule does not hold for general dynamic memberships}. This is supported by prior work \cite{arbitman2010backyard}, which shows that a dynamic exact membership problem only requires $nf'(0,\lambda)=(1+o(1))nf(0,\lambda)$ bits, and \cite{lovett2013space}, which demonstrates that a dynamic approximate membership problem costs $nf'(\epsilon,+\infty)=nC(\epsilon)$ $f(\epsilon,+\infty)$ bits, where $C(\epsilon)>1$ depends solely on $\epsilon$ (This implies that there exists a $\lambda>0$ for which $f'(\epsilon,\lambda)>f(\epsilon,\lambda)$). So the inequality
$$f'(0,\lambda)=f(0,\lambda)=f(\epsilon,\lambda)+f(0,\epsilon\lambda)< f'(\epsilon,\lambda)+f'(0,\epsilon\lambda)$$
serves as a counterexample to the chain rule (\bb{Theorem \ref{thm::chain rule}}).


\subsubsection{Replace the combining operator ``\&''}\label{algo::dsbe}
\

For the second limitation, we can design a space-efficient (i.e. $O(nf({0,\lambda}))$) exact filter with no additional construction space based on the chain rule. 

\textbf{Overview.} \text{Our key idea is} to replace the combining operator ``\&'' (i.e. {$\&_{i=1}^m\mathcal{F}_i$}) with ``$\&\sim$'' and recursively define \begin{equation*}\begin{cases}
    \mathcal{F}(\cdot):=\mathcal{F}^0(\cdot);\\
    \mathcal{F}^i(\cdot):=\mathcal{F}_{i+1}(\cdot)\&\sim \mathcal{F}^{i+1}(\cdot), \forall i\in \mathbb{N} \text{ (}\sim \text{means \texttt{NOT}}).
\end{cases}
\end{equation*}
In this formula, every elementary filter $\mathcal{F}_i(\cdot)$ is a approximate filter like Bloom Filter or Cuckoo Filter, and $\mathcal{F}_{i+1}$ is the whitelist of $\mathcal{F}_i$ (\textbf{Figure \ref{pic::algo2}}). To construct the filter, in step \textcircled{1}, we insert all positive items into $\mathcal{F}_1$; in step \textcircled{2}, we insert all false positive items of $\mathcal{F}_1$ into $\mathcal{F}_2$; in step \textcircled{3}, we insert all false positive items of $\mathcal{F}_2$ into $\mathcal{F}_3$, and so on. If we need an exact filter, this process iterates until no false positive items are left. It's easy to find that \algoname{} has no additional construction space, as all items can be placed on the input tape of the Turing Machine model.
\begin{figure}[h!tbp]
  \centering
   \includegraphics[width=0.48\textwidth]{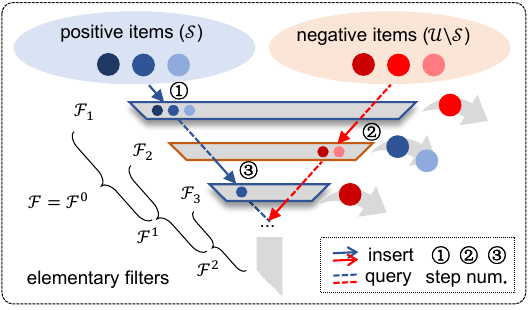}
\caption{The ``$\&\sim$'' version shown in Section \ref{algo::dsbe}.}
\label{pic::algo2}
\end{figure}

\textbf{Analysis.} Here we analysis the space cost of the new design (denoted as $nf^{CF}({0, \lambda })$ bits) and derive related parameter settings. 
\begin{theorem}\label{algo::variant} 
(\textbf{Space Cost}) Suppose an approximate filter costs $C'\log 1/\epsilon$ bits per item ($C'=\log e=1.44$ for Bloom Filter), then
    $$\inf f^{CF}({0, \lambda })=C'\log 4e\lambda.$$
\end{theorem}
\begin{proof}
Our proof is based on a simple observation: Exactly classifying $n$ positive items and $ \lambda  n$ negative items is equivalent to exactly classifying $\lambda\times n$ positive items and $\lambda \times(1/\lambda) n$ negative items. Formally speaking, we have $f({0,\lambda})=\lambda f({0,1/\lambda})$. 
We use the above formula to determine the parameter settings, and the following proof is only technically complicated. 

W.l.o.g. we let $\lambda\geqslant 1$. According to the chain rule
\begin{equation*}
\begin{split}
f(0,\lambda)=f(\epsilon,\lambda)+f(0,\epsilon\lambda)=f(\epsilon,\lambda)+\epsilon\lambda f(0,1/(\epsilon\lambda)),
\end{split}
\end{equation*}
we have
\begin{equation*}
\begin{split}
    \inf f^{CF}({0, \lambda })/C'&=\inf\limits_{\epsilon<1}\left\{ \log \frac{1}{\epsilon}+\epsilon  \lambda f^{CF}({0,\epsilon  \lambda })/C'\right\}\\
    &=\log  \lambda +\inf\limits_{\epsilon<1}\left\{\log \frac{1}{\epsilon  \lambda }+\epsilon  \lambda f^{CF}({0,\epsilon  \lambda })/C'\right\}\\
    &=\log \lambda + \inf\limits_{x<\lambda}\left\{ \log \frac{1}{x }+x \inf f^{CF}({0,\epsilon  \lambda })/C'  \right\}\\
    &=\log \lambda + \text{Constant}.
\end{split}
\end{equation*}
Since $\inf f^{CF}({0,\epsilon  \lambda })$ is derivable, so the right hand side achieves the minimum value when
\begin{equation*}
\begin{split}
     \frac{\mathrm{d}}{\mathrm{d}x}(\log\frac{1}{ x} + x \inf f^{CF}({0,\epsilon  \lambda })/C')=0
\end{split}  
\end{equation*}
\begin{equation*}
\begin{split}
     \text{Therefore }\begin{cases}
         \epsilon  \lambda =x = 1;\\
         \inf f^{CF}({0,\epsilon  \lambda })=C'\log  4e\lambda.
     \end{cases}
\end{split}  
\end{equation*}
Actually, according to L'Hôpital's rule, if
\begin{equation*}
\begin{split}
    &f^{CF}({0,\epsilon  \lambda })\equiv C'\left(\log  \lambda  +\log \frac{1}{\delta}+\delta f^{CF}({0,\epsilon  \lambda })\right)(\delta<1), \text{ we have}\\
    &f^{CF}({0,\epsilon  \lambda })=C'\left(\log  \lambda +\frac{1+\delta}{1-\delta}\log \frac{1}{\delta}\right)\to C'\log  4e\lambda (\delta \to 1).
\end{split}  
\end{equation*}
We show our algorithm in \textbf{Algorithm \ref{algo::DSBE}}.
\end{proof}

\renewcommand\arraystretch{1.4}
\begin{table*}[h!tbp]
	\centering
	\begin{tabular}{lcc@{\hskip 25pt}|ccc}
		\bottomrule
		\textbf{Properties}&\textbf{``\&'' version}&\textbf{``$\&\sim$'' version} & \textbf{Approximate}&\textbf{Exact}&\multirow{2}*{\textbf{Lower bound}}\\ 
		\textbf{($\epsilon=0$)}&(\textbf{Algorithm \ref{algo::construct}})&(\textbf{Algorithm \ref{algo::DSBE}})&\bb{Bloomier}& \bb{Bloomier}&~ \\
            \hline
		\textbf{Construction time}&$O(|\mathcal{U}| )$&$O(|\mathcal{U}| \log\log n)$&$O(n)$&$O(|\mathcal{U}|)$&Linear \\
            \hline
            \textbf{Filter space (bits)}&    $\approx Cn\log (2e\lambda\ln2)$ & $\approx  C'n\log 4e\lambda$ &$O(n\log \lambda n)$&$\approx C|\mathcal{U}|$&$n(\log \lambda+(\lambda+1)\log(1+1/\lambda))$ \\
            \hline
            \textbf{Additional space}&$\Omega(n\log n)$&/&/&$\Omega(|\mathcal{U}|\log|\mathcal{U}|)$&/ \\
		\toprule
	\end{tabular}
 \caption{Summary of \algoname{} variants}\label{algo::summary}
\end{table*}

\begin{algorithm}
\renewcommand\baselinestretch{0.9}\selectfont
\LinesNumbered
  \caption{Exact \algoname{} (using operator ``$\&\sim$'')}\label{algo::DSBE}
  \KwIn{Universe $\mathcal{U}$ and subset $\mathcal{S}$, $|\mathcal{U}|/|\mathcal{S}|=\lambda>1$.}  
  \KwOut{A filter $\mathcal{F}:\mathcal{U}\mapsto \{0,1\}$ s.t. $\mathcal{F}(e)=1$ iff $e\in S$.}
  \textbf{Function} Construct ($\mathcal{U},\mathcal{S}$):\\
  Set $\kappa =  \lambda , \delta \in (0,1),  \mathcal{S}_T=\mathcal{S},\mathcal{S}_F=\mathcal{U}\backslash\mathcal{S},i=1.$\\
  \textbf{While} $\mathcal{S}_F\neq \emptyset:$\\
  \quad\quad Construct an approximate filter $\mathcal{F}_i $ s.t. 
  \begin{equation*}
  \begin{cases}
  \mathcal{F}_i(e)=1, & \forall e\in \mathcal{S}_T;\\
  \mathbb{P}[\mathcal{F}_i(e)=0] \leqslant \delta/\kappa, & \forall e\in \mathcal{S}_F.
  \end{cases}
  \end{equation*}
  \\
  \quad\quad $\kappa \leftarrow 1/\delta, \mathcal{S}_T\leftarrow \mathcal{S}_F, \mathcal{S}_F\leftarrow \{e\in \mathcal{S}_F: \mathcal{F}_i(e)=0\}.$\\
  \quad\quad$i\leftarrow i+1$.\\
  \textbf{return} $\mathcal{F}(\cdot):=\mathcal{F}^0(\cdot)$, where $\mathcal{F}^j(\cdot):=\mathcal{F}_{j+1}(\cdot)\&\sim \mathcal{F}^{j+1}(\cdot), \forall j\in [0..i-2]$.
\end{algorithm}
\begin{Rmk}
    In practice, we can set $\delta=1/2$, round up the space cost of $\mathcal{F}_1$ to $C'n\lceil\log \lambda/\delta\rceil$ bits, and round up the space cost of $\mathcal{F}_i(i\geqslant 2)$ to $C'n2\delta^{i-1}\log 1/\delta=C'n2^{2-i}$ bits. The total space cost is no more than $C'n \log  16\lambda$ bits and the expected query time is $O(1)$. To take a step further, we can replace the last $O(\log n-\log\log n)$ approximate filters with one exact filter whose construction space is $O(n)$, so that we can reduce the number of filters from $O(\log n)$ to $O(\log \log n)$. We summarize the properties of the exact \algoname{} using operator ``\&'' and ``$\&\sim$'' in \textbf{Table \ref{algo::summary}} (``$/$'' means the additional space complexity is no more than the filter space complexity).  
\end{Rmk}

%

	\presec
\section{Applications and Evaluation}\label{perf}
\postsec

$\noindent\bullet$ \textbf{Experimental Setup}: In this section, we implement \algoname{} and its variants in C++ and Python and equip them with Murmur Hashing \cite{MurmurHash} to compute mapped addresses. We evaluate their performance in terms of space usage, speed and accuracy in several applications, including data compressing (\textbf{Section \ref{perf::compact truth table}, \ref{perf::randomized huffman coding}}), classifying (\textbf{Section \ref{perf::sah}, \ref{Learned Filter}}), and filtering (\textbf{Section \ref{perf::lsmtree}}). Note that the universe $\mathcal{U}$ may not be the absolute universe. Instead, it can be the set of frequently queried items (\textbf{Section \ref{perf::lsmtree}}). Unless otherwise stated, all item keys are 64-bit pre-generated random integers.  All experiments are conducted on a machine with 36 Intel$^\circledR{}$ Core\texttrademark{} i9-10980XE CPU @ 3.00GHz (576KiB L1 d-cache, 576KiB L1 i-cache, 18MiB L2 cache, 24.8MiB L3 cache) and 128GB DRAM.

$\noindent\bullet$ Our evaluation \textbf{metrics (with units)} are:

(a) \textbf{Filter space (Mb)}: The size of the filter measured in million bits (Mb). Additional construction space is not included.

(b) \textbf{Average construct and query throughput (Mops)}: The average number of operations per time, measured in million operations per second (Mops). Each experiment was repeated 10 times, and the mean value was recorded to reduce error.

(c) \textbf{Error rate}: The ratio of the number of misclassified (both false positive and false negative) items to the number of all items. We use this metric in \textbf{Section \ref{perf::sah}}.

(d) \textbf{Tail latency ($\mu s$)}: The high percentile latency measured in $\mu s$. For example, a P99 latency represents the time cost of an operation which is longer than 99\% time costs of all operations. We use this metric in \textbf{Section \ref{perf::lsmtree}}.

(e) \bb{False Positive Rate}: The ratio of the false positive items to the number of negative items. We use this metric in \textbf{Section \ref{Learned Filter}}.

\subsection{Static Dictionary}\label{perf::compact truth table}
As a warmup, in this part, we use \algoname{} ("\&" version, \textbf{Algorithm \ref{algo::construct}}) to compactly encode Boolean function $ \varphi: \{0,1,...,|\mathcal{U}|\}$ $\mapsto \{0,1\}$ with static support.

\subsubsection{\textbf{Modeling}} We regard all $n$ inputs which satisfy $\varphi(\vec{X})=1$ as positive items, and the other $\lambda n$ inputs as negative items. According to the \textbf{Remark} of \textbf{Theorem \ref{algo::optimality}}, the exact \algoname{} requires
\begin{equation*}
\begin{split}
\frac{C}{1+\lambda}\left(\lfloor\log \lambda\rfloor+1+\frac{\lambda}{2^{\lfloor\log \lambda\rfloor}}\right)|\mathcal{U}|\leqslant\frac{4C}{5\log 5-8} H(\frac{1}{\lambda+1}) |\mathcal{U}|\text{ bits,}
\end{split}
\end{equation*}
where $C<1.13$ according to the \textbf{Remark} of \textbf{Section \ref{algo::warm up}}. So we have
\begin{Cor}\label{Perf::Compact Truth Table}
When $C<1.13$, \algoname{} takes at most $4C/(5\log 5-8)=26\%$ space overhead to encode static dictionaries with high probability $1-o(1)$.
\end{Cor}
\begin{figure}[h!tbp]
  \centering
   \includegraphics[width=0.47\textwidth]{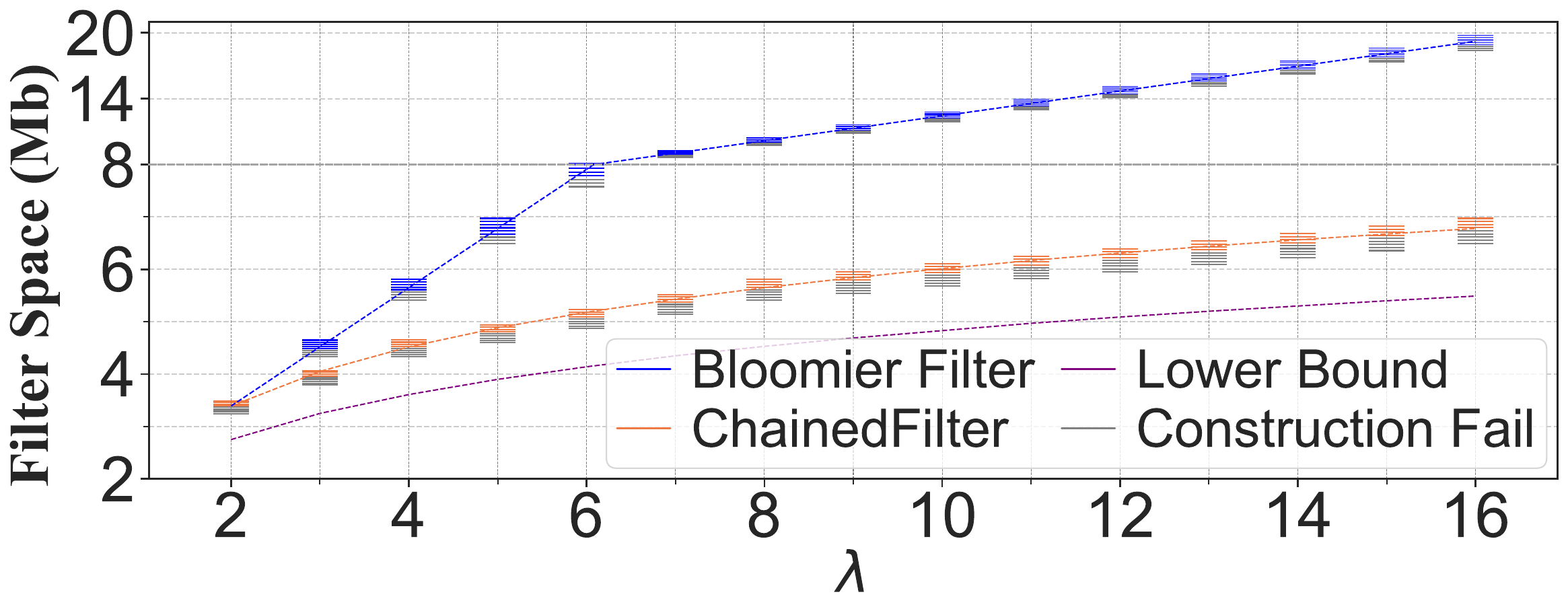}
\caption{Filter space of the exact Bloomier Filter, \algoname{}, and the theoretical lower bound. The dotted lines and the dashes represent theoretical and experimental results, respectively. We run each data point for 10 times with different hash seeds. A colored dash indicates a successful construction, while a grey dash indicates the construction fails (you may want to enlarge this figure).}
\label{pic::exp1}
\end{figure}
\subsubsection{\textbf{Experiments}} We fix $n=1$ million and vary the parameter $\lambda$ from 2 to 16 to compare the filter space, construction throughput and query throughput between \algoname{} and exact Bloomier Filter\footnote{When $\lambda=1$, \algoname{} exactly degenerates to exact Bloomier Filter.}. In \textbf{Figure \ref{pic::exp1}}, we find that the experimental space costs correspond well with our theory. Once the allocated space cost is larger than a certain threshold, the filter build processes will succeed with high probability. Specifically, when $\lambda=16$, \algoname{} costs 64\% less space than exact Bloomier Filter. In \textbf{Figure \ref{pic::exp23} (a)}, we discover that the construction throughput of exact Bloomier Filter decreases as $\lambda$ increases, whereas that of \algoname{} increases as $\lambda$ increases. This is because although a larger $\lambda$ leads to a worse locality, the throughput bottle neck of \algoname{} is from the space cost of exact Bloomier Filter, which fluctuates between $2Cn$ (when $\lambda=2,4,8,16$) and $3Cn$. Therefore, the amortized throughput of each item of \algoname{} increases, and it even surges from the case of $\lambda=2^i-1$ to the case of $\lambda=2^i (i=2,3,4,...)$. Specifically, when $\lambda =16$, the construction throughput of \algoname{} is 407\% higher than that of Exact Bloomier Filter. In \textbf{Figure \ref{pic::exp23} (b)}, we observe that the query throughput of \algoname{} shows the same phenomenon. This is because only the positive items and the false positive items require the lookup of the second stage filter. Specifically, when $\lambda =16$, the query throughput of \algoname{} is 103\% higher than that of Exact Bloomier Filter.
\begin{figure}[h!tbp]
  \centering
   \includegraphics[width=0.47\textwidth]{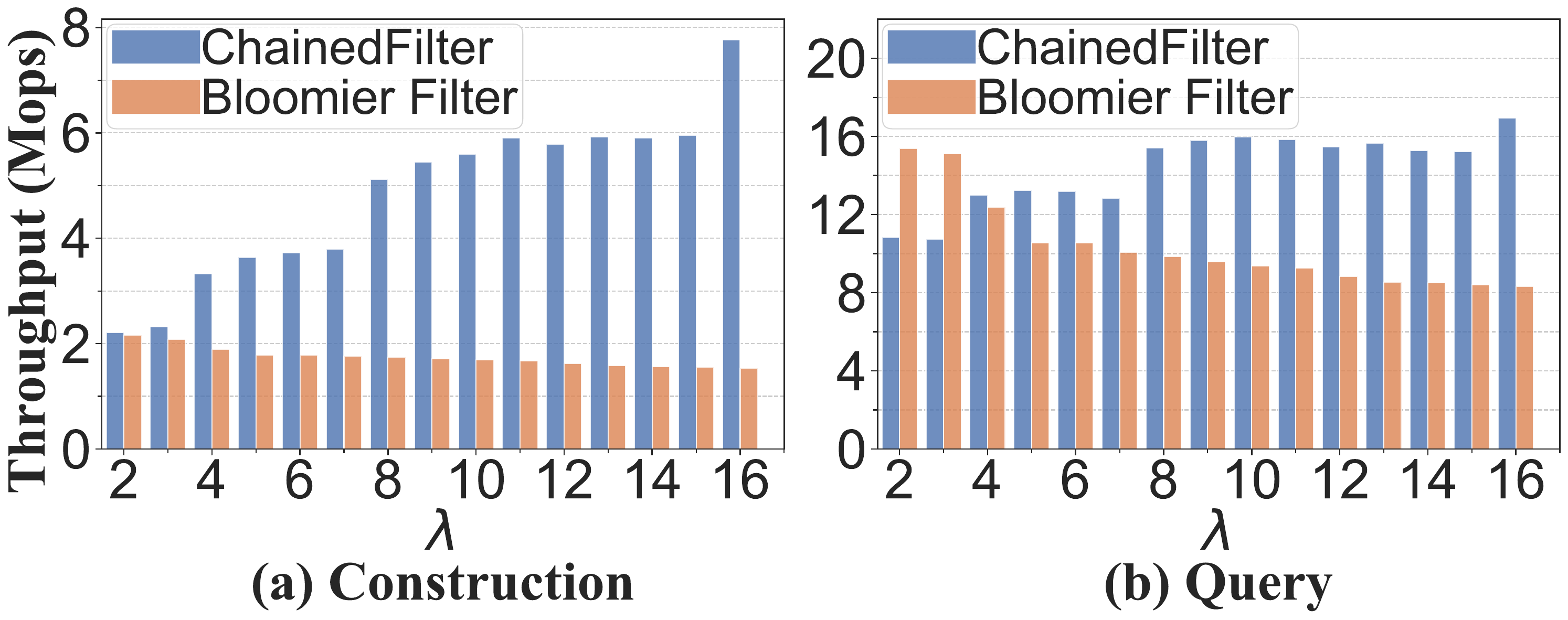}
\caption{Average construction and query throughput.}
\label{pic::exp23}
\end{figure}
\subsection{Random Access Huffman Coding}\label{perf::randomized huffman coding}
In this part, we apply \algoname{} to lossless data compression. 
\subsubsection{\textbf{Background}} In information theory, Huffman Coding \cite{huffman1952method} is an optimal prefix coding technique for lossless data compression. Given the possibility of occurrence of all kinds of symbols $\vec{p}:=(p_1,p_2,...)$, the Huffman's algorithm generates a variable-length code table, namely Huffman Tree, which assigns an $l_i-$bit prefix code for the $i$-th symbol. The literature shows that the average code length $\overline{{L}}_{\mathrm{Huff}}:=\sum l_ip_i$ satisfies 
$$H(\vec{p})\leqslant \overline{{L}}_{\mathrm{Huff}}< H(\vec{p})+1, \text{ where }H(\vec{p}):=-\sum\limits_i p_i\log p_i.$$
Although Huffman's algorithm is optimal for separate symbols, it has certain limitations: \textbf{(1) Compression ratio}. Unlike arithmetic coding or ANS \cite{duda2013asymmetric}, $\overline{{L}}_{\mathrm{Huff}}$ may not approach the entropy $H(\vec{p})$ arbitrarily. For example, if we have a string consisting of one character `a' and $1023$ character `b'. The Huffman's algorithm will cost $1024$ bits to encode them, but the least space cost is only 10 bits since we can simply record the address of `a'. \textbf{(2) Decoding order}. The Huffman's algorithm does not support random memory access of data. \textbf{(3) Confidentiality}. Attackers with prior knowledge of the character's frequency $\vec{p}$ can decipher the unencrypted Huffman Tree. \textbf{(4) Robustness}. In Huffman Code, small interference (bit flip or loss) may cause it to fail to recognize the starting position of characters, resulting in decoding failure. Some alternative approaches can partially address these problems. For example, \textbf{(1) blocking} some symbols can increase the compression ratio at the expense of the Huffman Tree's complexity; \textbf{(2) Encoding} the Huffman Code into a perfect hash table can support random access at the expense of space overhead \cite{hreinsson2009storing}; \textbf{(3) Encryption} and \textbf{(4) error correcting code} may increase the confidentiality and robustness at the expense of time and space overhead.

\subsubsection{\textbf{Modeling}} \algoname{} ("\&" version, \textbf{Algorithm \ref{algo::construct}}) can overcome the four limitations simultaneously. Given the Huffman Tree based on the probability vector $\vec{p}$, we encode every character's address paired with its Huffman Code into \algoname{}. Specifically, for the $i$-th character in the data with Huffman Code $\vec{v}=(v_1,v_2...,v_k)\in \{0,1\}^k$, we encode $(\text{key}=(i,j),\text{value}=v_j), j\in[1..k]$ into \algoname{}. To query the $i$-th character, we can query $\text{key}=(i,1),(i,2)...$ until the leaf node of the Huffman Tree. 

For example, to compress a string "\texttt{ab$\backslash$0}" with its Huffman Tree (`a'$\to $\texttt{00}, `\texttt{b}'$\to$ \texttt{01}, `$\backslash$\texttt{0}'$\to $\texttt{1}), we can encode the negative items $(1,1),$ $(1,2),$ $(2,1)$ and the positive items $(2,2), (3,1)$ into \algoname{}. To query, say, the second character, we first query key $=(2,1)$ and find that it is negative, and then query key$=(2,2)$ and find that it is positive. Therefore, we know that the second character whose Huffman Code is ``\texttt{01}'' is `\texttt{b}'. 

Our algorithm has random access property, high confidentiality (as long as the hash seed is secure), and high robustness. In \textbf{Theorem \ref{perf::Huff}}, we prove that \algoname{} optimizes the worst-case compression ratio performance as well.
\begin{theorem}\label{perf::Huff}
The average code length of our algorithm, denoted as $\overline{{L}}_{\mathrm{ours}}$, satisfies 
$H(\vec{p})<\overline{{L}}_{\mathrm{ours}}< H(\vec{p})+ 0.22$
 with high probability $1-o(1)$.
\end{theorem}
\begin{proof}
We only prove the right inequality by mathematical induction. For convenience, we define constant $C_0:=0.22$.
First, if the Huffman Tree has only two leaf nodes, according to the \textbf{Remark} of \textbf{Theorem \ref{algo::optimality}}, we have
\begin{equation*}
    \begin{split}
        \overline{L}_{\mathrm{ours}}&\leqslant \frac{C}{1+\lambda}\left(\lfloor\log \lambda\rfloor+1+\frac{\lambda}{2^{\lfloor\log \lambda\rfloor}}\right)\\
        &\leqslant H(p)+C+\frac{2}{3}-\log 3<H(p)+C_0.
    \end{split}
\end{equation*}
Consider there are two Huffman Trees that satisfy the above inequality. We combine them into a larger Huffman Tree with weights $q$ and ($1-q$), respectively. Since the number of layers is increased by one, we have
\begin{equation*}
\begin{split}
    \overline{L}&_{\mathrm{ours}}=\sum\limits_iqp_{1i}(l_{1i}+1)+\sum\limits_i(1-q)p_{2i}(l_{2i}+1)\\
    &\leqslant qH(\vec{p}_1) + (1-q)H(\vec{p}_2) + 1 + C_0\\
    &=-\sum\limits_i qp_{1i}\log qp_{1i}-\sum\limits_i (1-q)p_{2i}\log (1-q)p_{2i}\\
    &\quad+q\sum\limits_i \log q +(1-q)\sum\limits_i \log (1-q) + 1 +C_0\\
    &=H(\vec{p}) -(H(q)-1)+ C_0<H(\vec{p}) +C_0.
\end{split}
\end{equation*}
So $\overline{L}_{\mathrm{ours}}<H(\vec{p})+C_0$ holds for all Huffman Trees.
\end{proof}
\begin{Rmk} \textbf{Theorem \ref{perf::Huff}} demonstrates that our algorithm has a tighter upper bound (0.22 bit per item) on space overhead than the standard Huffman's algorithm (1 bit per item). This is particularly advantageous when the data is highly skewed or sparse.  However, one drawback of \algoname{} is its poor spacial locality, which limits its throughput in hierarchical memory systems. To compensate this drawback, we can enable the first and the second stage filters to share the same mapped addresses, as suggested in \cite{reviriego2021approximate}. The optimized version uses a $Cn\lceil\log\lambda\rceil$-bit approximate Bloomier Filter (with less than $n$ false positive items) and a $2Cn$-bit exact Bloomier Filter, and every $(\lceil\log\lambda\rceil+2)$ bits are organized as a block. With this optimization, each item has $j=3$ common mapped blocks shared by both of the first and the second stages, making it possible to access its value within $j$ memory accesses.
\end{Rmk}
\begin{figure}[h!tbp]
  \centering
   \includegraphics[width=0.48\textwidth]{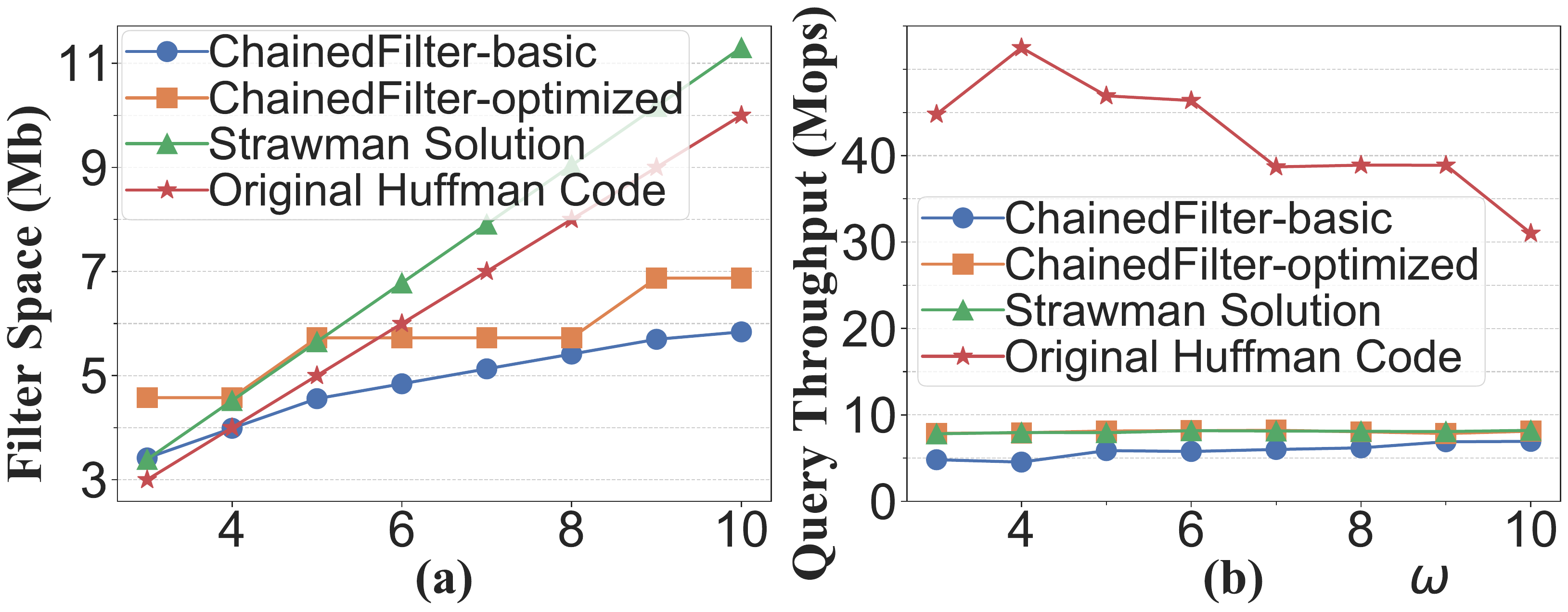}
\caption{Filter space and query throughput of random access Huffman coding algorithms.}
\label{pic::huffman}
\end{figure}
\subsubsection{\textbf{Experiments}} We synthesize eight datasets of strings whose characters' occurrence obeys an exponential distribution with parameter $\omega=3,4,...,10,$ respectively\footnote{For example, when $\omega=3$, the dataset may have (expected) 1 character `\texttt{a}', 3 character `\texttt{b}', 9 character`\texttt{c}', and 27 character `\texttt{d}', etc.}. We fix the number of positive items to 1 million, generate the Huffman Code for each string, and encoding the Huffman Code using \algoname{} and its optimized version (\textbf{Remark} of \textbf{Theorem \ref{perf::Huff}}). To provide a strawman solution, we encode the Huffman Code into an exact Bloomier Filter and record its filter space and query throughput. For reference, we also evaluate the performance of raw Huffman Coding, but present its sequential decoding throughput rather than random decoding throughput (because it does not support random access). The experimental results are shown in \textbf{Figure \ref{pic::huffman}}, we find that \algoname{} performs better as $\omega$ grows. When $\omega=10$, the basic and the optimized versions of \algoname{} saves 48.3\% and 39.2\% of the filter space compared to the strawman solution, while their query throughput is only 15.5\% and 1.17\% slower.
\subsection{Self-Adaptive Hashing}\label{perf::sah}
In this part, we use \algoname{} to reduce the number of memory accesses required for Cuckoo Hashing \cite{pagh2004cuckoo}.  
\subsubsection{\textbf{Background}} Serving as a hash predictor, \algoname{} reduces unnecessary memory accesses for multiple-choice hashing, which take around $100 ns$ for DRAM and $\approx 150 \mu s$  for NAND SSD. Instead, it requires only a few fast (around $10$ ns) in-cache lookups. 

Based on "the power of two random choices" \cite{mitzenmacher2001power}, Cuckoo Hashing consists of two hash tables $T_1[1..M]$ and $T_2[1..M]$, allowing an item $e$ to have two potential mapping locations $T_1[\mathrm{h}_1(e)]$ and $T_2[\mathrm{h}_2(e)]$. Each item can only occupy one location at a time, but can swap between the two if its prior location is taken. To insert an item $e$, we first attempt to place it into $T_1[\mathrm{h}_1(e)]$. If $T_1[\mathrm{h}_1(e)]$ is already occupied by another item $e'$, we evict $e'$ and reinsert it into $T_2[\mathrm{h}_2(e')]$. We repeat this process until all items stabilize. If this is impossible, we reconstruct the cuckoo hash table with a new hash seed. The literature shows that the insertion failure probability is $O(1/M)$ when the occupancy (load factor $r$) is less than $1/2-\varepsilon$. However, to query an item, we often check both of the two hash tables, which can result in significant latency penalties. 

A well known solution to reduce the number of memory accesses for Cuckoo Hashing is to use a pre-filter, such as a Bloom Filter \cite{bloom1970space} or Counting Bloom Filter\footnote{The Counting Bloom Filter replaces every bit to a counter to support deletion of existing items. It is called Counting Block Bloom Filter \cite{manber1994algorithm} if we further restrict all mapped bits in the same block.} \cite{fan2000summary}, to predict the mapped locations of each item. However, false predictions by the pre-filter can cause additional memory accesses. To address this problem, EMOMA \cite{pontarelli2018emoma} adds a 1(block):1(bucket) pre-Counting Block Bloom Filter that corresponds to the first hash table, and locks certain items in the second hash table to prevent movements that may cause false positives. This creates an always-exact hash predictor that supports line-rate processing (e.g. in programmable switches or FPGAs) at the expense of complex insertion process and space overhead.
\begin{figure}[h!tbp]
  \centering
   \includegraphics[width=0.48\textwidth]{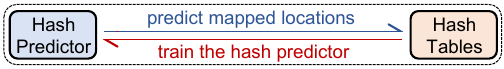}
\caption{The concept of the self-adaptive hashing.}
\label{pic::sah}
\end{figure}
\subsubsection{\textbf{Modeling}} Instead of supporting absolute exact matches, we use \algoname{} ("$\&\sim$" version, \textbf{Algorithm \ref{algo::DSBE}}) as a pre-filter to predict the mapped locations with best effort. \text{Our key idea is} to let false predictions train the predictor to reduce errors (\textbf{Figure \ref{pic::sah}}). Specifically, we regard items in the first hash table as negative and those in the second hash table as positive. To query an item $e$, we ask \algoname{} for the predicted location. If the prediction is incorrect, we adjust \algoname{} by flipping the mapped bits in $\mathcal{F}_1,\mathcal{F}_2,\mathcal{F}_3,...$ to 1 until \algoname{} can accurately predict the locations of $e$. In the \textbf{Remark} of \textbf{Theorem \ref{algo::variant}}, we have proved that the error rate will ultimately converge to zero as long as the space cost of is greater than $C'n\left(\log  \lambda -((1+\delta)/(1-\delta))\log\delta\right)     
 $. Next we show how to choose the negative-positive ratio $\lambda$. 
\begin{theorem}\label{perf::ratio}
Given the expected number of items $|\m{U}|$ and the number of buckets in the Cuckoo hash table, $2M$, we define $r:=|\m{U}|/2M<1/2-\varepsilon$ is the load factor. Then we have 
    $$\lambda=\left(\frac{2r}{1-e^{-2r}}-1\right)^{-1}+o(1).$$
\end{theorem}
\begin{proof}
    Let $\eta$ be the number of negative items (items in the first hash table). A new item can either insert into an empty location in the first hash table, increasing $\eta$ by one, or evict an old item with probability $\eta/M$ (assuming no insertion failures). This problem can be reformulated as the well-known coupon collector's problem, where $M$ is the number of coupons and $\eta$ is the number of draws. Therefore, we have
    \begin{equation*}\begin{split}
        2r=\sum\limits_{l=M-\eta}^{M}\frac{1}{l} + o&(1) =\int\limits_{1-\eta/M}^{1}\frac{1}{x}\mathrm{d}x+o(1)=\ln\left(1-\frac{\eta}{M}\right)+o(1)\\
        &\Rightarrow \lambda:=\frac{\eta}{2rM-\eta}=\left(\frac{2r}{1-e^{-2r}}-1\right)^{-1}+o(1).
    \end{split}
    \end{equation*}
\end{proof}
\begin{Rmk}
According to \textbf{Theorem \ref{perf::ratio}} and the \textbf{Remark} of \textbf{Section \ref{algo::dsbe}}, when $\delta=1/2$, \algoname{} incurs a space cost of no more than
$C'n\log 16\lambda=2C'r/(\lambda+1)\log 16\lambda\cdot M$ bits. In contrast, the Counting Block Bloom Filter of EMOMA \cite{pontarelli2018emoma} costs $8M$ bits if every block has two 4-bit counters. In the following experiments, we fix the Cuckoo hash table size $2M$ to 1 million. \textbf{Table \ref{perf:lambda}} shows the space cost (measured in Mb) of \algoname{} ($\delta=1/2$) and EMOMA when $r\in[0.1,0.4]$ (for reasonable comparison, the space cost of Bloom Filter equals to that of EMOMA). We observe that \algoname{} is much more space-efficient: it saves $76.7\%$ (r=0.4) $\sim$ $99.75\%$ (r=0.1) space of EMOMA to predict exact locations. 
\end{Rmk}
\renewcommand\arraystretch{1}
\begin{table}[h!tbp]
	\centering
	\begin{tabular}{c|c@{\hskip 7pt}c@{\hskip 7pt}c@{\hskip 7pt}c@{\hskip 7pt}c@{\hskip 7pt}c@{\hskip 7pt}c@{\hskip 7pt}c}
		\bottomrule
            \textbf{\diagbox{Space}{$r$}} & \textbf{0.10} &\textbf{0.15}&\textbf{0.20}&\textbf{0.25}&\textbf{0.30}&\textbf{0.35}&\textbf{0.40}\\
            \hline
             \textbf{EMOMA} & 4.00 &4.00&4.00&4.00&4.00&4.00&4.00\\
            \textbf{\algoname{}} & 0.10 &0.20&0.32&0.45&0.60&0.76&0.93\\
            \toprule
	\end{tabular}
\caption{Filter space of EMOMA and \algoname{}.}
\label{perf:lambda}
\end{table}
\subsubsection{\textbf{Experiments}} 
%
%
%
%
%
In the first experiment, we show that although \algoname{} may not initially predict the mapped locations absolutely exactly, its error rate decreases exponentially and can quickly converge to zero. To verify this, we set $r=0.4$ and query all items in order for $R$ rounds to train the \algoname{} and show the change in error rate in \textbf{Figure \ref{pic::errorrate} (a)}. We find that only $0.34\%$ items are wrongly predicted after three rounds of training, and all items are exactly predicted after seven rounds of training, resulting in a reduction of $(\lambda+1)^{-1}|_{r=0.4}=31\%$ external memory access for Cuckoo Hashing. To further accelerate the training process, we insert all items of the $\lfloor\log\log n\rfloor=4$-th layer $\mathcal{F}_4$ into an Othello hash table \cite{yu2018memory} according to the \textbf{Remark} of \textbf{Section \ref{algo::dsbe}}, so that the training process can be completed in four rounds (``ChainedFilter-optimized'' in \textbf{Figure \ref{pic::errorrate}}). In the second experiment, we use keys of length 36 bytes and values of length 64 bytes and compare the throughput of raw Cuckoo Hashing, Cuckoo Hashing with \algoname{}, with EMOMA, and with Bloom Filter (\textbf{Figure \ref{pic::errorrate} (b)}). We find that although \algoname{} uses only $23.3\%$ of the filter space of EMOMA, its average construction and query throughput is 17\% and 41\% faster, respectively.  It is interesting to note that both Cuckoo Hashing with EMOMA and Cuckoo Hashing with Bloom Filter exhibit slower query throughput compared to raw Cuckoo Hashing, possibly due to their higher hash computation overhead.
\begin{figure}[h!tbp]
  \centering
   \includegraphics[width=0.48\textwidth]{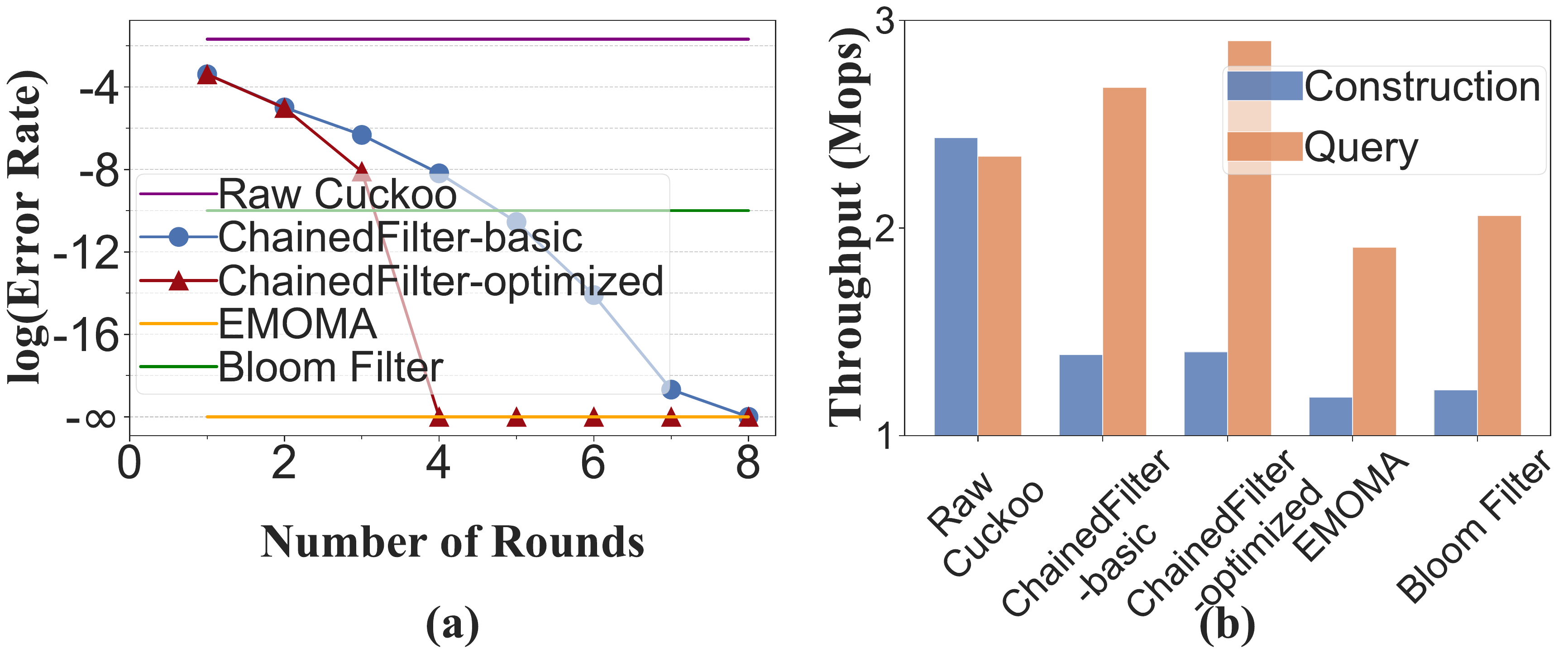}
\caption{Error rate and throughput of raw Cuckoo Hashing, Cuckoo Hashing with \algoname{} (0.93Mb), with EMOMA (4Mb), and  with Bloom Filter (4Mb).}
\label{pic::errorrate}
\end{figure}
\subsection{Point Query for LSM-Tree}\label{perf::lsmtree}
In the final application, we utilize \algoname{} to reduce the tail point query latency of LSM-Tree.
\subsubsection{\textbf{Background}}
LSM-Tree (short for log-structured merge-tree) is a storage system for high write-throughput of key-value pairs. It uses an in-memory data structure called memtable to buffer all updates until it is full, and then flushes the contents into the persistent storage as a sorted run through a process called minor compaction. However, since the sorted runs may have overlapping key ranges, LSM-Tree has to check all of them for a point query, which can result in poor query performance. To address this problem, LSM-Tree uses a hierarchical merging process called major compaction to merge sorted runs\footnote{The most commonly used compaction strategies are leveled compaction and tiered compaction. The leveled compaction minimizes space amplification at the expense of read and write amplification, while the tiered compaction minimizes write amplification at the cost of read and space amplification. Here we only focus on the tiered compaction.}. In tiered major compaction, the LSM-Tree is organized as a sequence of levels, and the number of runs in each level is bounded by a threshold $T$. If the number achieves the threshold, the compaction process is triggered to merge all $T$ sorted runs a new sorted run (SSTable) in the next level. To further speed up point queries, each SSTable typically has an approximate filter, such as a Bloom Filter, to skip most non-existing items. However, since the approximate filter has false positives, in the worst case, a point query still needs to check all SSTables, resulting in poor tail point query latency.
%
%
%
%
\begin{figure}[h!tbp]
\centering
\includegraphics[width=0.48\textwidth]{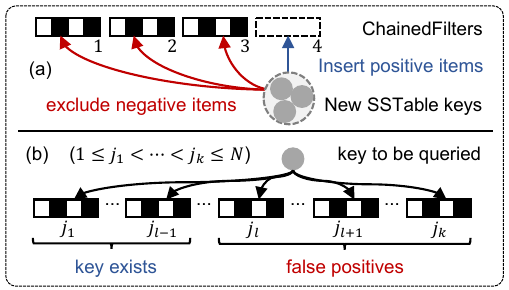}
\caption{\algoname{} (``\& version'') for LSM-Tree. Note that the second stage part must be a dynamic filter.}
\label{pic::rocksdb}
\end{figure}
\subsubsection{\textbf{Modeling}} To reduce the latency, we replace every approximate filter with a dynamic exact \algoname{} (``\&'' version, \textbf{Algorithm \ref{algo::construct}} and \textbf{Section \ref{algo::replace elementary filters}}). Suppose there are $N$ SSTables in one level, our goal is to reduce the worst case additional SSTable search time from $N$ to 1 (whether or not there are repeating items). Our key idea is, for \algoname{} of the $i$-th SSTable in this level, we regard all keys in the $i$-th SSTable as positive items, and regard all other keys in the $(i+1),(i+2),...,N$-th SSTables (but not in the $i$-th SSTable) as negative items (\textbf{Figure \ref{pic::rocksdb} (a)}). In this way, an exact \algoname{} says ``yes'' only when (1) the queried key is in its corresponding SSTable, (2) the queried key is not in the subsequent $(i+1),(i+2),...,N$-th SSTables in the same level (otherwise the key must be excluded when the later SSTables are formed). Consider the case where the $j_1<j_2<...<j_k$-th \algoname{}s report ``yes''. Our strategy is to check the corresponding SSTables in order until we find a false positive SSTable (note that the item keys are repeatable). Once we detect that the $j_l$-th \algoname{}'s result is a false positive, we can assert that all the $j_{l+1},j_{l+2},...,j_{k}$-th \algoname{}'s results are false positives as well (\textbf{Figure \ref{pic::rocksdb} (b)}). Therefore, the number of additional SSTable searches in this level is no more than one.
\subsubsection{\textbf{Experiments}} We measure the tail point query latency of our approach in RocksDB database \cite{dong2021rocksdb} by replacing the built-in Bloom Filter with the \algoname{}. To simplify the implementation, we let all SSTables stay in the first level (\texttt{L0}) and use the default 64MB write buffer size\footnote{RocksDB stores temporal write operations in memtable and flushes it to disk to generate an SSTable file in \texttt{L0} when its size reaches the write buffer size.}. We reuse the built-in Bloom Filter as the first stage of \algoname{}, but additionally implement an Othello hash table as the dynamic second stage filter. When creating a new SSTable after a compaction, we generate the Bloom Filter as what RocksDB does, but additionally query its keys in prior SSTables in the same level and update the false positive SSTables' \algoname{}s by including the false positive items into the second stage filter (\textbf{Figure \ref{pic::rocksdb} (a)}). When querying an item, we check the positive SSTables (whose Bloom Filter reports ``yes'') in order. Once we find a false positive one, we can assert that all later possible SSTables are false positives (\textbf{Figure \ref{pic::rocksdb} (b)}).

\begin{figure}[h!tbp]
  \centering
   \includegraphics[width=0.47\textwidth]{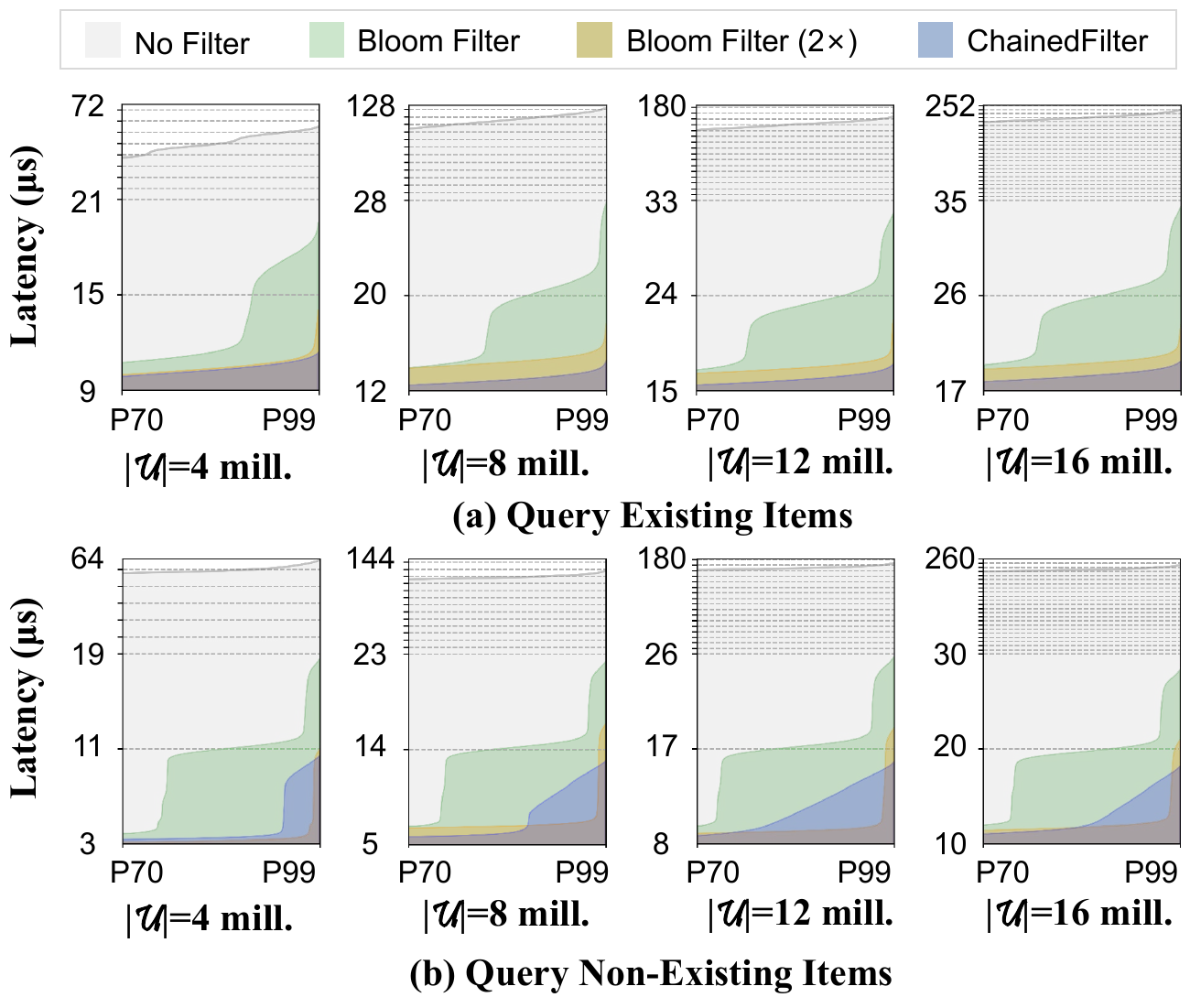}
\caption{Tail point query latency in RocksDB.}
\label{pic::lsm}
\end{figure}

\textbf{Figure \ref{pic::lsm}} demonstrates that \algoname{} can significantly reduce the tail point query latency of LSM-Tree. We generate items with distinct 36-Byte random keys and 64-Byte values and vary the number of items $|\mathcal{U}|$ from 4 million (400MB data, 7 SSTables in total) to 16 million (1.6GB data, 30 SSTables in total). The filter space of \algoname{} just allows it to exactly classify positive and negative items. We use RocksDB with only Bloom Filters with the $0\times, 1\times, \text{ and }2\times$ filter space as the comparison algorithm (``$0\times$'' means no filter). In \textbf{Figure \ref{pic::lsm} (a)}, we show the tail latency of querying existing items\footnote{Since all existing keys are different, we end the query process once we find an item.}. Let us consider the ``$|\mathcal{U}|$\textbf{=16 mill.}'' sub-figure as an example. The green area in the figure shows the tail point query latency using only built-in Bloom Filter which costs the same space as \algoname{}. The \texttt{P0-P77}, \texttt{P77-P95}, and \texttt{P95-P99} tail latency approximately represents the query latency with zero, one, and more than one false positive SSTable reads. In the worst case, with many false positive SSTable reads, the \texttt{P99} tail point query latency is about $31\mu s$. In contrast, as shown in the blue area,  the query process using \algoname{} has no false positive when querying existing items, so the tail latency remains under $20\mu s$, which is 36\% lower. In \textbf{Figure \ref{pic::lsm} (b)}, we show the tail latency of querying non-existing items. We find that when $|\mathcal{U}|$ is 16 million, the latency with \algoname{} gradually increases. This is because the latency is not only determined by the number of false positive SSTables, but is also influenced by the index of the false positive \algoname{}. For example, if the first SSTable reports a false positive, the algorithm can quickly return, and thus the overall latency may be even lower than $12\mu s$. However, if the 30-th SSTable is the first one to report false positive, the overall latency will be more than $18\mu s$.

\subsection{Learned Filter}\label{Learned Filter}
In this part, we use \algoname{} to reduce the false positive rate of Learned Filters.
\subsubsection{\bb{Background}}  Although our space lower bound (\bb{Section \ref{Space Lower Bound}}) assumes that the positive items are randomly  drawn from the universe, further compression can be achieved when the items follow a specific data distribution. Learned Filters \cite{kraska2018case,mitzenmacher2018model,liu2020stable,dai2020adaptive} incorporate a continuous function, such as an RNN model, in front of the filter structure to capture the data distribution. When querying an item, if the continuous function outputs ``yes'', the Learned Filter immediately reports true. However, when the continuous function outputs ``no'', the item is sent to a backup Bloom Filter to eliminate false negatives. In the learned filter structure, both the learning model and the backup filter may introduce false positives, but in some cases, the overall false positive rate even decreases.

\subsubsection{\bb{Experiments}} Extending the chain rule (\bb{Theorem \ref{thm::chain rule}}) to general membership problems with different data distributions is an intriguing problem\footnote{An analogy is, some machine learning algorithms like Boosting \cite{freund1999short} combine multi-stage elementary classifiers to work together, which seem similar to \algoname{}.}. Since we haven't derive an elegant theory at this moment, we empirically replace the backup Bloom Filter / Bloomier Filter with our \algoname{} (``\&'' version) to show the experimental improvements. We refer to the open-source code of the learned Bloom Filter and the dataset on GitHub \cite{LearnedFilter}. The dataset has 30,000 positive (good) and 30,000 negative (bad) websites evaluated by users. In this experiment, we randomly select different proportions (ranging from $0\%$ to $100\%$) of data to train the RNN model. Generally, the model's generalization ability improves as the amount of training data increases. Experimental results show that when we fix the overall false positive rate at 0.01\footnote{For Learned \algoname{} and Learned Bloomier Filter, we set the false positive rate of the RNN model to 0.01, and the false positive rate of the backup filter to zero. For Learned Bloom Filer, we set the false positive rate of both the RNN model and the backup filter to 0.005.}, the filter space of Learned \algoname{} can be up to 99.1\% lower than that of Learned Bloom Filter and 96.2\% lower than $nf(\epsilon,\lambda)$ (the space lower bound without considering data distribution).

\begin{figure}[h!tbp]
  \centering
  \includegraphics[width=0.48\textwidth]{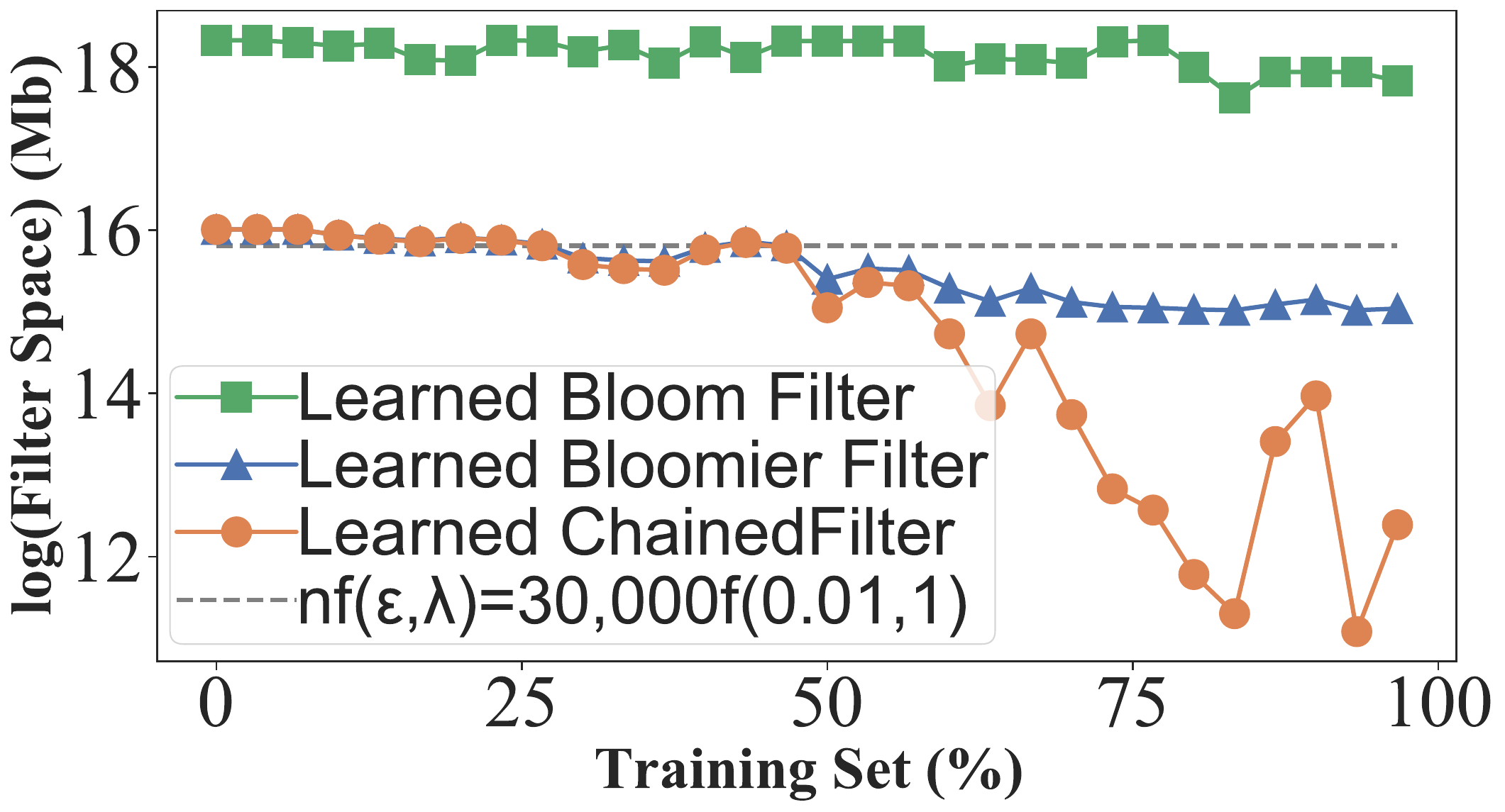}
\caption{Filter space in bits ($\log$ scale) of Learned Bloom Filter, Learned Bloomier Filter and Learned \algoname{}.}
\label{pic::errorrate}
\end{figure}



\presec
\section{Related Work}\label{sec:related}
Filter algorithms are the foundation of many important problems \cite{wu2021elastic,9944968,stingysketch,li2023ladderfilter,guo2023sketchpolymer,Rencoder,namkung2022sketchlib,li2023steadysketch,fan2022evolutionary,namkung2023sketchovsky,liu2023hypercalm} and have been well studied. In this section, we introduce more membership filters that may work as our elementary filters. 
\subsection{Approximate Membership Filters}
\postsec
Approximate filters originated with Bloom Filter \cite{bloom1970space} in 1970. A Bloom Filter is a bitmap of size $m$ where every one of $n$ positive items is mapped using $k$ independent hash functions. To insert an item, we set all mapped positions to one; To query an item, we check whether all mapped positions are one, and report positive iff they are. A Bloom Filter satisfies one-sided error: a query is either “definitely not” (no false negative) or “probably yes” (small false positive). The false positive rate depends on the space cost $$\epsilon=\left(1-\left(1-\frac{1}{m}\right)^{nk}\right)^k\approx\left(1-e^{-\frac{nk}{m}}\right)^k\Rightarrow m\geqslant \frac{n\log1/\epsilon}{\ln2}.$$ 
In 1978, \cite{carter1978exact} gave the tight space lower bound ($\log 1/\epsilon+o(1)$ bits per item) for approximate filters, which means the Bloom Filter wastes no more than $1/\ln2-1=44\%$ space. In 2010, \cite{lovett2013space} gave the space lower bound $C(\epsilon) \log 1/\epsilon$ bits per item for dynamic approximate data structures, where $C(\epsilon)>1$ depends only on $\epsilon$. Many later works have emerged to reduce the space overhead of Bloom Filter. Cuckoo Filter  (2014) \cite{fan2014cuckoo} borrows the concept from Cuckoo Hashing (\textbf{Background} of \textbf{Section \ref{perf::sah}}) and uses fingerprints for approximate classification. It maps an item $e$ to two buckets $h(e)$ and $(h(e)\oplus f(e))$ in one hash table, and shows that the load factor can be up to 95\% if the buckets have four slots. Therefore, the space cost drops to $1.05(2+\log1/\epsilon)$ bits per item. Inspired by Bloomier filter and the peeling theory, XOR Filter  (2019) \cite{graf2020xor} and Binary Fuse Filter  (2022) \cite{graf2022binary} achieves a space cost of $C\log1/\epsilon$ bits per item for static membership query (\textbf{Remark} of \textbf{Section \ref{algo::warm up}}). In fact, a more compact space cost ($(1+o(1))\log1/\epsilon$ bits per item) is given by \cite{dietzfelbinger2008succinct} (2008) and \cite{porat2009optimal} (2009) for static membership query if we allow a more complex implementation.

\subsection{Exact Membership Filters}

One type of exact filters builds upon perfect hashing \cite{majewski1996family}. The Bloomier Filter \cite{chazelle2004bloomier,charles2008bloomier} (2004) is a milestone which requires $1.23$ bits per item (the theoretical result is excerpted from IBLT (2011) \cite{goodrich2011invertible}). According to \cite{walzer2021peeling} (2021), this constant can be optimized to close to 1 (\textbf{Remark} of \textbf{Section \ref{algo::warm up}}). In fact, a more compact space cost of $(1+o(1))$ bits per item is achievable if a more complex implementation is allowed \cite{dietzfelbinger2008succinct,porat2009optimal}. Othello Hashing \cite{yu2018memory} (2016) and Coloring Embedder \cite{tong2019coloring} (2021) provide dynamic perfect hashing designs with $O(1)$ bits per item. However, they both need additional structures with $\Omega(n\log n)$ space for construction and update. 

Another one type of exact filters is called ``dictionaries''. In 1984, \cite{fredman1984storing} described a general constant-time hashing scheme (FKS dictionary) with a space complexity of $O(n)$ words (not bits). In 1994, researchers designed an static exact filter with a space complexity of $O(B)$ \cite{brodnik1994membership} (they later improved to $B(1+O(1/\log\log\log|\mathcal{U}|))$ \cite{brodnik1999membership} in 1999) for static dictionary, where $B:=\lceil\log\binom{|\mathcal{U}|}{n}\rceil$ is the lower bound for static membership problems \cite{carter1978exact}. In 2001, the space complexity is optimized to nearly optimal ($B+O(\log\log (|\mathcal{U}|))+o(n)$ bits) \cite{pagh2001low}. In 2010, \cite{arbitman2010backyard}
 proofs the lower bound $(1+o(1))B$ also holds for dynamic exact membership filters. Nowadays, the research on exact membership problems is still on going \cite{bercea2020space}.

In the end, we list some other membership filters satisfying special properties. In 2005, \cite{pagh2005optimal} started considers dynamic general membership problems and designed a membership filter with a space complexity of $((1 + o(1))n\log 1/\epsilon +
O(n+\log|\mathcal{U}|)$ bits for arbitrary $\epsilon$ and $\lambda$. Since 2018, researchers have creatively introduced machine learning methodology to construct Learned Bloom Filters \cite{kraska2018case,mitzenmacher2018model,liu2020stable,dai2020adaptive}, or add pre-filters to construct membership algorithms \cite{reviriego2021approximate}, which are more accurate than the original Bloom Filter.


\presec
\section{Conclusion}
\label{sec:conclusion}
\postsec

In this paper, we present a lossless factorization theorem, namely chain rule, for solving general membership query problems. Based on this theorem, we propose a simple yet space-efficient filter framework called \algoname{} and apply it to various applications. Both theoretical and experimental results show that the \algoname{} outperforms its elementary filters. We believe our chain rule can inspire more innovative works, and our \algoname{} can be used in more practical applications.

\section*{Acknowledgement}
We thank to Maria Apostolaki (Princeton University) and Jiarui Guo (Peking University) for valuable comments of this work.
 	{
 	    \bibliographystyle{unsrt}
 	    \balance
 \bibliography{InputFiles/reference}
 	}

\end{document}